\documentclass[a4paper]{article}

\pdfoutput=1

\usepackage[utf8]{inputenc}
\usepackage[T1]{fontenc} 
\usepackage{amsthm}
\usepackage{hyperref}
\usepackage{color}
\usepackage{amsmath}
\usepackage{amssymb}
\usepackage{dsfont}
\usepackage{physics}
\usepackage{mathtools}
\usepackage{float}
\usepackage{algpseudocodex}
\algrenewcommand\algorithmicprocedure{$\triangleright$}

\usepackage{booktabs}

\usepackage[a4paper, margin=1.5in]{geometry}
\usepackage{authblk}

\usepackage{tikz, graphics}       %
\usetikzlibrary{arrows.meta,      %
                bending,          %
                patterns,         %
                shapes,
                positioning}

\usepackage [
  n,
  advantage,
  operators,
  sets,
  adversary,
  landau,
  probability, 
  notions,
  logic,
  ff,
  mm,
  primitives,
  events,
  complexity,
  asymptotics,
  keys
  ] {cryptocode}

\usepackage{graphicx}
\graphicspath{{Figures}{./Figures/}{../Figures/}}

\newcommand{\Abort}{\mathsf{Abort}}
\newcommand{\Fail}{\mathsf{Fail}}
\newcommand{\Succ}{\mathsf{Success}}

\newcommand{\Acc}{\mathsf{Accept}}

\newcommand{\BQP}{\mathsf{BQP}}

\newcommand{\X}{\mathsf{X}}
\newcommand{\Y}{\mathsf{Y}}
\newcommand{\Z}{\mathsf{Z}}

\newcommand{\CZ}{\mathsf{CZ}}

\newcommand{\cptp}[1]{\mathsf{#1}} %

\newcommand{\Pois}{\mathrm{Poisson}}
\newtheorem{theorem}{Theorem}
\newtheorem{lemma}{Lemma}
\newtheorem{remark}{Remark}
\newtheorem{corollary}{Corollary}

\floatstyle{ruled}

\newfloat{resource}{htb!}{idf}
\floatname{resource}{Resource}

\newfloat{simulator}{htb!}{idf}
\floatname{simulator}{Simulator}

\newfloat{protocol}{htb!}{idf}
\floatname{protocol}{Protocol}

\newfloat{routine}{htb!}{idf}
\floatname{routine}{Routine}

\newfloat{game}{htb!}{idf}
\floatname{game}{Game}

\newfloat{algo}{htb!}{idf}
\floatname{algo}{Algorithm}

\title{Composably Secure Delegated Quantum Computation with Weak Coherent Pulses}

\date{}

\author[1]{Maxime Garnier}
\author[2]{Dominik Leichtle}
\author[3]{Luka Music}
\author[1]{Harold Ollivier}
\affil[1]{DIENS, École Normale Supérieure --  PSL Research University, CNRS, INRIA, 45 rue d'Ulm, 75005 Paris, France}
\affil[2]{School of Informatics, University of Edinburgh, 10 Crichton Street, Edinburgh EH8 9AB, United Kingdom}
\affil[3]{Quandela, 7 rue Leonard de Vinci, 91300 Massy, France}

\begin{document}

\maketitle

\begin{abstract}
  A client can delegate a quantum computation to a powerful remote server while ensuring the privacy and the integrity of its computation via Secure Delegated Quantum Computation (SDQC). Thanks to recent results making them noise-robust and resource-efficient, proofs-of-concept implementations of generic SDQC protocols have already been demonstrated. Yet, the requirements for implementing them are still too stringent to go beyond this step while maintaining high security expectations.

  To further reduce their physical resource cost, we show how to realise SDQC using weak coherent pulses (WCPs) instead of single photons. More precisely, we construct a protocol which guarantees that, among a sufficiently large batch of transmitted WCPs, at least one contained only a single photon. This holds even if the adversary controls the transmittance of the photonic link connecting the client and the server. Our protocol's security is proven in the composable Abstract Cryptography (AC) framework.

  This batch can then be fed to known quantum privacy amplification techniques to prepare a single secure qubit in the $\X - \Y$ plane, which can be used in any composable SDQC protocol which relies on the secure preparation of single qubits. Furthermore, the guarantee on the batch of states can also be used for Quantum Key Distribution (QKD) where the privacy amplification step is classical. In doing so, we address a weakness in the standard security proof of the decoy state method.

  While our protocol can be instantiated with any number of different intensities for the WCPs, using only two intensities already shows improved scaling at low transmittance, thus opening the possibility of increasing the distance between the client and the server.
\end{abstract}

\vspace*{0.5cm}
\begin{footnotesize}
	\noindent This paper extends the results presented in~\cite{GLMO24composably}. In particular, it enhances the scaling at low transmittance from \(1/\eta^3\) to \(1/\eta^2\), showcases numerical results and provides the full proofs of all technical statements. Furthermore, this work clarifies the relationship with the standard QKD decoy-state method by uncovering unrecognised weaknesses in the security proof and solving them. We also explicitly provide arguments for the intuition behind the protocol as well as the proofs.
\end{footnotesize}

\section{Introduction}

Building trust in quantum computing devices that are remotely available through service providers and vendors is of utmost importance. As they start reaching sizes of a few hundred qubits~\cite{C23ibm}, they approach the utility regime where they will hold sensitive data and valuable algorithmic developments. It is thus timely to develop methods to ensure the confidentiality and integrity of the end-user's computations.

One of the key questions underlying the research in Secure Delegated Quantum Computation (SDQC) concerns the possibility of classically verifying quantum computations~\cite{G04conference,A07scott,AV13is}. On the one hand, it is known that it can be achieved by a fully classical client at the price of computational hardness assumptions~\cite{M18classical}. On the other hand, statistically secure schemes exist but require the client to be able to send single-qubit states to the server \emph{i.e.} have quantum capabilities. Thus investigating the possible reduction in quantum resources on the client side of known statistically secure SDQC protocols is an important question for both practical implementability and foundational reasons.

Following this line of research, new verification protocols for Bounded-error Quantum Polynomial-time ($\BQP$) computations were introduced \cite{LMKO21verifying,KKLM22unifying} to add noise robustness to existing statistically secure SDQC protocols. Crucially, the overhead consists only of a polynomial number of computations that are not more difficult to implement than the unprotected one. These features made it possible to implement a proof-of-concept with trapped ions~\cite{DNMN23verifiable}. Additional simplifications have later been developed to further improve their practicality, such as in~\cite{KLMO24verification} which removes the assumptions of trusted preparations and measurements on the client side and shows that $\Z(\theta)$ rotations and bit flips are sufficient to obtain the same security guarantees. Furthermore, the security of this family of protocols is guaranteed in the stringent Abstract Cryptography framework~\cite{MR11abstract-cryptography}, which entails sequential and parallel composability, i.e.~the ability to reuse these protocols as building blocks for larger protocols while preserving their security properties. This approach enables modularity in both the theoretical and the experimental developments.

To further these efforts, this paper relaxes the assumption that the honest client sends single qubits. From a practical perspective, these qubits will most likely be of photonic nature. Hence, we replace the requirement of single-photon sources by simpler and cheaper Weak Coherent Pulse (WCP) sources. Using pulses of various intensities in a way similar to the decoy-state method~\cite{LMC05decoy}, we show that is it possible to send a batch of quantum states to the server with the guarantee that at least one of them was emitted as a single-photon state in a composably secure way. We prove that the correctness and security errors are negligible in the number of pulses sent. This technique is then composed with a quantum privacy amplification gadget~\cite{KKLM23asymmetric} to recover one secure single-photon state per batch. These states are then used in the protocols from~\cite{KKLM22unifying,KKLM23asymmetric} to construct an SDQC protocol.

\paragraph{Contributions and Related Work}

This works builds on~\cite{DKL11universal} where SDQC is realised with WCPs of a single intensity. Our work extends it by adding verifiability and a better scaling of the number of required pulses at low transmittance $\eta$. More specifically, the simplest instantiation with only two intensities achieves a $1/\eta^2$ scaling for the number of sent pulses compared to $1/\eta^4$ in~\cite{DKL11universal}. Compared to the protocol and implementation of remote state preparation in~\cite{JWHX19remote}, our work adds composability, which was a missing piece to claim security of their blind (but not verifiable) quantum computation protocol with decoy states. Independently from this work,~\cite{TM24unconditional} propose a verification protocol with WCPs that relies on~\cite{DKL11universal} and thus retains the $1/\eta^4$ scaling while also lacking composability.

In the course of this work, we uncovered a weakness in the security proof of decoy state Quantum Key Distribution (QKD) as it implicitly assumes that the number of pulses emitted with a given number of photons within a block of pulses has indeed converged toward its average value. While this is legitimate for low photon number pulses, this clearly breaks down for higher photon numbers. This jeopardizes the full security proof of this widely used method together with its finite-size analysis, and thereby possibly affects the scaling for low transmittance.

Fortunately, the composable security of our protocol makes it applicable to QKD, as it relies on the same batch preparation of quantum states before applying classical privacy amplification techniques. As our security proof does not rely on the precise shape of the photon number distribution, as do those for the decoy state method (e.g.~\cite{LCWX14concise}) our technique also offers robustness against WCP source imperfections.

\section{Preliminaries}
\label{sec:preliminaries}

An \emph{MBQC pattern} consists of a graph $G = (V, E)$, two subsets $I$ and $O$ of $V$ defining input and output vertices, a list of angles $\{\phi_v \}_{v \in V}$ with $\phi_v \in \Theta \coloneqq \{k\pi/4\}_{0 \leq k \leq 7}$ and a flow function $f : O^c \rightarrow I^c$. The computation starts by preparing the graph state $\ket{G} = \prod_{(v, w) \in E} \CZ_{v, w} \bigotimes_{v \in V}\ket{+}_v$. Then each qubit $v$ of $V$ is measured along the basis $\ket{\pm_{\phi'_v}}$ in an order defined by the flow, with $\ket{+_\alpha} = (\ket 0 + e^{i\alpha}\ket 1)/\sqrt 2$. The angle $\phi'_v$ depends on $\phi_v$, the outcomes of previously measured qubits and the flow. Any $\BQP$ algorithm in the circuit model can be translated in the MBQC model with at most a polynomial overhead (See e.g.~\cite{DKP07measurement-calculus}). Such patterns can be blindly delegated to a server thanks to the Universal Blind Quantum Computation (UBQC) Protocol~\cite{BFK09universal}. It replaces the initial $\ket+$ states by $\ket{+_\theta} $ with $\theta\sample \{k\pi/4\}_{0 \leq k \leq 7}$ and further adapts the measurement angles to compensate for this modified initialisation. Then, using blindness, SDQC protocols have been constructed~\cite{FK17unconditionally} by inserting traps (i.e. qubits whose measurement outcome is deterministic if performed correctly and known only to the client) that probe the behaviour of the server and allow the client to reject potentially corrupted results. By replacing local traps by interleaved computation and test rounds (rounds whose outcomes are deterministic and easy to compute for the client), \cite{LMKO21verifying,KKLM22unifying} provided protocols which have negligible correctness and security errors, no memory overhead and can additionally withstand a fixed amount of global noise.

In this paper, we use the protocol from~\cite{KKLM23asymmetric}, derived from these earlier works, which requires only $\ket{+_{\theta}}$ states and define a construction dubbed \(\mathsf{BatchRSP}\) allowing to send a batch of $\ket{+_{\theta}}$ states, one of which is guaranteed to be unknown to the server. We then use the Collaborative Remote State Preparation construction from~\cite{KKLM23asymmetric}, to transform the batch of states into protocol a single such state with perfect security. These recombined states can then be used in their protocol performing SDQC
\begin{theorem}[SDQC from Batch Remote State Preparation~\cite{KKLM23asymmetric}, informal]
  \label{thm:asymmetric}
  The protocol from~\cite{KKLM23asymmetric} uses \(\mathsf{BatchRSP}\) as a primitive to perform SDQC for $\mathsf{BQP}$ with information-theoretic composable security and negligible security error.
\end{theorem}

Our security is defined in the \emph{Abstract Cryptography (AC)} framework~\cite{MR11abstract-cryptography}. A protocol is secure if it is a good approximation of an ideal and secure-by-design \emph{resource}. This should hold if all participants follow the protocol, but also when malicious ones deviate from it arbitrarily. Secure protocol in this framework are inherently composable: sequential or parallel executions of secure protocols are also secure. Intuitively, if protocol $\Pi_A$ $\epsilon$-approximates resource $A$ and protocol $\Pi_B$ uses $A$ to approximate resource $B$, then replacing the call to $A$ in $\Pi_B$ by $\Pi_A$ degrades the security of the composed protocol by at most $\epsilon$. For this reason, Abstract Cryptography is a more stringent security framework compared to non-composable standalone security criteria.

\section{Modelling Weak Coherent Pulses}
\label{sec:wcpg}

A weak coherent pulse (WCP) with intensity $\mu > 0$ sent through an optical setup with transmittance $\eta$ results in the state $\sum_n \frac{(\eta \times \mu)^n}{n!}e^{-(\eta \times \mu)} \rho^{\otimes n}$ being created at the output of the setup. Here, $\rho$ is a single photon density matrix, and we have assumed phase randomisation.

To formally analyse the security of protocols relying on WCPs, we introduce the $\mathsf{WCPGenerator}$ (Resource~\ref{res:wcp}) that captures the properties of a source of WCPs. A malicious receiver of the WCPs may chose --- via a flag bit $c$ --- to remove the losses of the optical channel, thereby increasing the probability of multiphoton states. For each multiphoton event, the receiver instead receives the classical description of the state, simulating the leak of all information regarding the quantum state.\footnote{Resource~\ref{res:wcp} gives more power to the Receiver compared to a real WCP source. Hence, a protocol secure with Resource~\ref{res:wcp} is also secure with a real WCP source.} In the following, we denote $\cptp U(\rho) = \cptp U \rho \cptp U^\dagger$.

\begin{figure}[ht]
\begin{resource}[H]
  \caption{$\mathsf{WCPGenerator}$}
  \label{res:wcp}
  \begin{algorithmic}[0]
    \State \textbf{Public information:} set of unitaries $\mathcal{U}$; classical description $[\rho_0]$ of a quantum state, transmittance $\eta \in [0,1]$.
    \State \textbf{Sender's Input:} $\cptp U \in \mathcal{U}$ and $\mu \in \mathbb{R}^+_0$.
    \State \textbf{Server's Input:} $c\in\bin$, set to $0$ if honest.
    \Procedure{\textbf{Computation by the Resource}}{}
    \If{$c=0$}
    \State It samples $n \sample \Pois(\mu\times\eta)$ and sends the state $\cptp U(\rho_0)^{\otimes n}$ to the Receiver, with $\cptp U(\rho_0) = \cptp U \rho_0 \cptp U^\dagger$.
    \Else
    \State It samples $n \sample \Pois(\mu)$. If $n\leq 1$, it sends the state $\cptp U(\rho_0)^{\otimes n}$ to the Receiver. If $n>1$, it sends $n$ and the classical description of $\cptp U$ to the Receiver.
    \EndIf
    \EndProcedure
  \end{algorithmic}
\end{resource}
\end{figure}

\begin{remark}\label{rem:robustness}
  All our results generalise straightforwardly to distributions other than Poisson, so long as the malicious distribution allows the adversary to emulate the honest behaviour. This impacts only the instantiation of the estimation strategy in Section~\ref{sec:example}.
\end{remark}

\section{Batch State Preparation from Weak Coherent Pulses with Composable Security}
\label{sec:batch-rsp}
An important building block for several high level protocols such as QKD or SDQC is Remote State Preparation (RSP). RSP allows a party to remotely prepare inside another party's register a single copy of an unknown quantum state drawn from a non-trivial set without revealing any additional information about the state's identity. As an immediate consequence, WCPs are not suited for RSP since they may contain several copies of the state to be prepared, thereby potentially leaking information about the prepared through the additional copies. To work around this inherent incompatibility, we propose a new resource called \(\mathsf{BatchRSP}\) (Resource \ref{res:batchrsp}) that guarantees that at least one in a batch of transmitted quantum states has been sent as a single copy. Privacy amplification gadgets can then be used to recover directly useful cryptographic functionalities: a classical one for QKD and the Collective RSP gadget from~\cite{KKLM23asymmetric} for SDQC. The states prepared are of the form $\cptp U_i (\rho_0)$ for a fixed known state $\rho_0$ and $\cptp U_i \in \mathcal{U}$. For technical reasons, we require the set of unitaries $\mathcal{U}$ to be a group closed under the Quantum One-Time Pad, i.e.~$\X,\Z \in \mathcal{U}$, which is not a problem in practical cases.

\begin{figure}[ht]
\begin{resource}[H]
  \caption{$\mathsf{BatchRSP}$}
  \label{res:batchrsp}
  \begin{algorithmic}[0]
    \State \textbf{Public Information:} $K \in \mathbb{N}^+$; set of unitaries $\mathcal{U}$; classical description $[\rho_0]$ of a quantum state.
    \State \textbf{Sender's Input:} $K$ unitaries $\cptp U_1,\ldots, \cptp U_K \in \mathcal{U}$.
    \State \textbf{Receiver's Input:} $c \in\bin$, set to $0$ if honest.
    \Procedure{\textbf{Computation by the Resource}}{}

    \If{$c=0$}
    \State It sends the states $\cptp U_1 (\rho_0), \ldots, \cptp U_K (\rho_0)$ to the Receiver.
    \Else
    \State It samples $k \sample \{ 1,\ldots,K \}$ uniformly at random.
    \State It sends $k$, the classical descriptions of $\{ \cptp U_i \}_{i\neq k}$, and the quantum state $\cptp U_k(\rho_0)$ to the Receiver.
    \EndIf
    \EndProcedure
  \end{algorithmic}
\end{resource}
\end{figure}

At a high level, Protocol~\ref{proto:decoy-state-method} constructs the \(\mathsf{BatchRSP}\) resource from WCPs in
three rounds.
(i) The first round consists in the generation and transmission of $N$ WCPs with at most \(N\) different
intensities. While all the possible values for the intensities are known to the Receiver, the order is permuted so that they do not know which pulse is sent with which intensity.
(ii) Upon receiving each pulse, the Receiver performs a non-destructive measurement to identify vacuum pulses and discards them. Among the remaining pulses, it chooses a subset of size $K$ from which it extracts a single photon -- where $K$ is the size of the batch --  and returns the corresponding indices to the Sender. If there are not enough such pulses, the Receiver aborts.
(iii) The final round consists of the Sender runs an estimation Algorithm~\({\cal B}\) on the indices of the pulses accepted by the Receiver to
decide whether to continue. If this estimation phase succeeds, the Sender computes the corrections and communicates them to the Receiver who applies them. These corrections ensure that the final states correspond to the Sender's desired states.

The estimation Algorithm~\({\cal B}\) is left unspecified in the protocol and our result (Theorem~\ref{thm:decoy-realization}) reduces the protocol's correctness and security in the composable AC framework to simple conditions on this algorithm. These can then be seen as generic guidelines for designing secure protocols, while allowing for enough flexibility to be applicable to a multitude of settings (in particular, different choices for the pulse intensities). It makes its decision based on correlations between the intensities and the accepted pulses. For correctness sake, if the transmittance of the channel is $\eta$, it must accept batches generated by an honest Sender and Receiver. To preserve security on the other hand, in the worst malicious case with transmittance equal to $1$, it must reject if it senses that the Receiver is trying to make the Sender accept a batch that contains only pulses with two or more photons.

Informally, a key insight into the security of Protocol~\ref{proto:decoy-state-method} is to realise that the Sender explicitly selects the order of the WCPs at random before sending them while never divulging this order. This randomness is precisely what allows us to restrict the strategy of the cheating Receiver to counting the number of photons received in each pulse and deciding whether or not to accept it\footnote{This fact is often used in other works based on WCPs such as decoy state methods in QKD, but without a proper justification.}. It cannot use the link between the pulse intensity and the number of photons to make its decision, meaning that our estimation algorithm can use this hidden information to detect whether the Receiver attempted to cheat. %

\begin{figure}[htp]
\begin{protocol}[H]
  \caption{Multi-Intensity Weak Coherent Pulse Method}
  \label{proto:decoy-state-method}
  \begin{algorithmic}[0]
    \State \textbf{Public Information:} $N, K\in \mathbb{N}^+$, where $K \leq N$; group of unitaries $\mathcal{U}$; classical description $[\rho_0]$ of a quantum state; transmittance $\eta \in [0,1]$; pulse intensities $\mu_1, \dots, \mu_N \in \mathbb{R}^+_0$; efficient, classical estimation algorithm $\mathcal{B}$.

    \State{\textbf{Sender's Input:}} $K$ unitaries $\cptp U_1,\ldots, \cptp U_K \in \mathcal{U}$.

    \Procedure{\textbf{Sender - Sending WCP}}{}
    \State Sample a random permutation $\pi \sample \operatorname{S}_N$.
    \State Let $(\mu_1', \ldots, \mu_N') \gets (\mu_{\pi(1)}, \ldots, \mu_{\pi(N)})$.%
    \For{$i \in [\![1,N]\!]$}
    \State Sample $\cptp U_i' \sample \mathcal{U}$ from the Haar measure on $\mathcal{U}$.
    \State Call $\mathsf{WCPGenerator}$ with $\cptp U = \cptp U_i'$, $\mu = \mu_i'$, $[\rho_0]$, and transmittance $\eta$.
    \EndFor
    \EndProcedure

    \Procedure{\textbf{Receiver - Acknowledging WCP Reception}}{}
    \State Let $\tilde{I}$ be the set of rounds whose pulses contain at least one photon.

    \If{$|\tilde{I}| < K$}
    \State Send $\Abort$ to the Sender and stop.
    \EndIf
    \State Let $I' \gets$ random subset of $\tilde{I}$ of size $K$.
    \State For $i \in I'$, store one round-$i$ photon in quantum memory.
    \State Send $I'$ to the Sender.
    \EndProcedure

    \Procedure{\textbf{Sender - Estimation}}{}
    \State Undo the permutation, i.e.~let $I \gets \{ \pi^{-1}(i) | i \in I' \}$.
    \State Run estimation algorithm $\mathcal{B}$ on input $I$.
    \If{$\mathcal{B}$ returns $\Abort$}
    \State Send $\Abort$ to the Receiver and stop.
    \EndIf
    \State Relabel the set of kept states in $I'$ from $1$ to $K$. Sample a random permutation $\sigma \sample \operatorname{S}_K$.
    \For{$j \in [\![1,K]\!]$}
    \State Let $\tilde{\cptp{U}}_j = \cptp{U}_{\sigma(j)} \cptp{U'}_j^{\dagger}$.
    \EndFor
    \State Send the corrections  $(\tilde{\cptp{U}}_j)_{j\in \tilde{I}}$ and the permutation $\sigma$ to the Receiver.

    \EndProcedure

    \Procedure{\textbf{Receiver - Corrections}}{}
    \For{$j \in [\![1,K]\!]$}
    \State Apply the unitary $\tilde{\cptp{U}}_j$ to the $j$-th kept state.
    \EndFor
    \State Permute all kept $K$ quantum states using $\sigma^{-1}$, and set the result as the output.
    \EndProcedure
  \end{algorithmic}
\end{protocol}
\end{figure}

\begin{figure}[ht]
\begin{game}[H]
  \caption{Correctness of the Multi-Intensity Weak Coherent Pulse Method}
  \label{game:correctness}
  \begin{algorithmic}[0]
    \State \textbf{Parameters:} $N,K \in \mathbb{N}^+$, where $K \leq N$; $\mu_1, \dots, \mu_N \in \mathbb{R}^+_0$; $\eta \in [0,1]$; algorithm $\mathcal{B}$.
    \Procedure{GameCor}{}
    \State $\tilde{I} \gets \{ i=1,\ldots,N | n_i \sample \Pois(\mu_i \times \eta), n_i \geq 1 \}$.
    \If{$|\tilde{I}| < K$}
    \State Return $\Abort$ and stop.
    \EndIf
    \State Sample at random a subset $I$ of $\tilde{I}$ of size $K$.
    \State Run $\mathcal{B}$ on input $I$ and return its output.
    \EndProcedure
  \end{algorithmic}
\end{game}
\end{figure}

\begin{figure}[ht]
\begin{game}[H]
  \caption{Security of the Multi-Intensity Weak Coherent Pulse Method}
  \label{game:security}
  \begin{algorithmic}[0]
    \State \textbf{Parameters:} $N,K \in \mathbb{N}^+$, where $K \leq N$; $\mu_1, \ldots, \mu_N \in \mathbb{R}^+_0$; algorithm $\mathcal{B}$.
    \Procedure{GameSim}{$\mathcal{D}$}

	\State Sample $\{n_i \sample \Pois(\mu_i)\}_{i \in [1,N]}$.
    \State $J_n \gets \{ i \in [1,N] | n_i = n \}$, $c_n \gets | J_n |$ for $n \in \mathbb{N}$.
    \State $(d_n)_{n \in \mathbb{N}} \gets \mathcal{D} \left( (c_n)_{n \in \mathbb{N}} \right)$, with $d_n \leq c_n \forall n \in \mathbb{N}$.
    \State $I_n \gets$ random subset of $J_n$ of size $d_n$, such that $I \gets \bigcup_{n\in\mathbb{N}} I_n$ satisfies $|I| = K$.
    \If{$\mathcal{B}(I)$ returns $\Acc$, and $d_0 = d_1 = 0$}
    \State Return $\Fail$.
    \Else
    \State Return $\Succ$.
    \EndIf
    \EndProcedure
  \end{algorithmic}
\end{game}
\end{figure}

Theorem~\ref{thm:decoy-realization} reduces the composable correctness and security of the protocol to entirely classical conditions (Game~\ref{game:correctness} and~\ref{game:security}) while remaining completely oblivious to how the pulse intensities are chosen and how the estimation is performed.
This means that $\mu_1, \ldots, \mu_N$ and $\mathcal{B}$ can be optimised independently, so long as they satisfy the sufficient conditions for correctness and security.

\begin{theorem}\label{thm:decoy-realization}
\newcounter{count:decoy-realization}
\setcounterref{count:decoy-realization}{thm:decoy-realization}
  For given transmittance $\eta \in [0,1]$, parameters $N,K \in \mathbb{N}^+$, $\mu_1, \ldots, \mu_N \in \mathbb{R}^+_0$, and algorithm $\mathcal{B}$, Protocol~\ref{proto:decoy-state-method} $\varepsilon$-statistically constructs Resource~\ref{res:batchrsp} ($\mathsf{BatchRSP}$) from Resource~\ref{res:wcp} ($\mathsf{WCPGenerator}$) if Game~\ref{game:correctness} returns $\Abort$ with probability at most $\varepsilon$, and Game~\ref{game:security} returns $\Fail$ with probability at most $\varepsilon$ for all maps $\mathcal{D}$.
\end{theorem}

Informally, Game~\ref{game:correctness} identifies the two possibilities for aborting the protocol and hence computing the correctness error, namely not enough pulses with more than one photon where received or the estimation procedure has aborted. Game~\ref{game:security} does the same for the security of Protocol~\ref{proto:decoy-state-method}.
In a nutshell, it formalizes the information available to the (malicious) Receiver and its optimal attack strategy: acknowledging $K$ pulses with no single-photon pulse while passing the estimation phase given by algorithm $\mathcal{B}$ and only knowing the photon numbers of each pulse. The detailed proof of Theorem~\ref{thm:decoy-realization} is given in Appendix~\ref{app:proof-thm-protocol}. By the composition of AC protocols, we can combine Theorem~\ref{thm:decoy-realization} with the result from~\cite{KKLM23asymmetric} (Theorem~\ref{thm:asymmetric}) to construct from WCPs an SDQC protocol with negligible error for classical input $\mathsf{BQP}$ computations.

\begin{corollary}
  There exists a WCP-based noise-robust SDQC protocol for $\mathsf{BQP}$ with information-theoretic composable security and negligible security error.
\end{corollary}

\section{An Instantiation with Two Intensities}
\label{sec:example}
In this section, we provide the simplest non-trivial instantiation of our technique by using two different pulse intensities:
half of the $N$ pulses will be emitted with intensity $\nu$ and the other half will be emitted with intensity $\nu'$, where $0< \nu < \nu' $ without loss of generality. We note $P$ the number of pulses with intensity $\nu$ that are acknowledged by the Receiver as having at least one photon. Similarly $P'$ is the number of acknowledged pulses reported by the Receiver with intensity $\nu'$. In the language of Protocol~\ref{proto:decoy-state-method}, $P = |\{i \in I, \mu_i = \nu\}|$ and $P' = |\{i \in I, \mu_i = \nu'\}|$ in the honest case. Both quantities can be evaluated by the Sender as it keeps track of the intensity $\mu_i$ used for the $i$-th pulse. The Sender then performs the following steps.
\begin{enumerate}
\item It computes
\begin{equation}
  T = \frac{1}{bc'-b'c}\left(c'P - cP' \right),\label{eq:T}
\end{equation}
with $a = e^{-\nu}, b = \nu e^{-\nu}, c = \nu^2 e^{-\nu}/2$ and similarly for $a', b', c'$ where $\nu$ is replaced by $\nu'$.
\item It computes $t$, the average value of $T$ that would have been obtained with an honest Receiver:
\begin{equation}
 t = \frac{1}{bc'-b'c}\left( c'(1-e^{-\eta\nu}) - c(1-e^{-\eta\nu'}) \right)\frac{N}{2}.\label{eq:t}
\end{equation}
\item It accepts whenever $T > t - \Delta_0 \frac{N}{2}$ for a constant $\Delta_0 > 0$.
\end{enumerate}
To keep the presentation of this section concise, the motivation for these choices is deferred to Section~\ref{sec:discussion}.

The procedure above fully defines Algorithm~$\mathcal B$\begin{figure}[ht]
\begin{algo}[H]
  \caption{Estimation algorithm $\mathcal B$}
  \label{algo:b}
  \begin{algorithmic}[0]
    \State \textbf{Input:} $I$
    \Function{Estimation}{}
    \State Fix $\Delta_0 > 0$.
    \State Compute $T$ from Eq.~\ref{eq:T} and $t$ from Eq.~\ref{eq:t}
    \EndFunction
    \If{$T  \geq t - \Delta_0 \frac{N}{2}$}
    \State Return $\Acc$
    \Else
    \State Return $\Abort$
    \EndIf
  \end{algorithmic}
\end{algo}
\end{figure}

\begin{theorem}[Correctness and Security Errors]
\label{thm:instantiation}
\newcounter{count:instantiation}
\setcounterref{count:instantiation}{thm:instantiation}
  Given $\Delta_0 > 0$ in Algorithm~$\mathcal B$, $\delta > 0$ such that $K = ((2-e^{-\eta\nu}-e^{-\eta\nu'})/2 - \delta)N$ and $C = \max(c, c')$, the correctness error $\varepsilon_{\text{corr}}$ satisfies
  \begin{align}
    \varepsilon_{\text{corr}} \leq \exp(-\delta^2 N) + \exp(-\frac{\Delta_0^2 (bc' - b'c)^2}{4 C^2}N)
  \end{align}
while the security error is negligible in $N$ whenever there are additional constants $\Delta_0', \Delta_0'' > 0$ such that
\begin{align}
\label{eq:uuubound}
&\Delta_0 + \Delta_0' + \frac{c'}{bc'-b'c}\Delta_0'' \\
&= \frac{c'(1-e^{-\eta\nu}) - c(1-e^{-\eta\nu'}) - c'(1-a-b-c)}{bc'-b'c}. \nonumber
\end{align}
\end{theorem}

\begin{proof}[Correctness error --- Proof sketch]
  Following Theorem~\ref{thm:decoy-realization}, the correctness error is upper bounded by the probability of Game~\ref{game:correctness} to return $\Abort$. This can happen if $|I|<K$ or if Algorithm~$\mathcal B$ returns $\Abort$. The probability of each event is upper bounded separately. The corresponding bounds rely on the following steps: (i) define the random variable $P = D_0 + D_1 + D_2 + D_{\geq 3}$ where $D_i$ (resp. $D_{\geq i}$) is the number of acknowledged pulses with intensity $\nu$ with $i$ photons (resp. $\geq i$ photons). Likewise, define $P'=D_0' + D_1' + D_2' + D_{\geq 3}'$ for pulses with intensity $\nu'$. (ii) Note that $P$ (resp. $P'$) is binomially distributed with parameter $N/2$ and $1-e^{-\eta\nu}$ (resp. $N/2$ and $1-e^{-\eta\nu'}$), so that $2P/N$ and $2P'/N$ will concentrate exponentially fast with $N$ around their average using Hoeffding's bound. (iii) $I$ defined in Game~\ref{game:correctness} satisfies $|I| = P + P'$ as an honest Receiver should never acknowledge a zero-photon pulse (i.e. $D_0 = D_0' =0$) and should acknowledge all other pulses.
\end{proof}

\begin{proof}[Security error --- Proof sketch]
  Following Theorem~\ref{thm:decoy-realization}, the security error $\varepsilon_{\text{sec}}$ is upper bounded by the probability that Game~\ref{game:security} returns $\Fail$. With the notation used earlier, this happens whenever $D_0 + D_0' + D_1 + D_1' = 0$ (i.e.~the Receiver never acknowledges pulses with less than 2 photons) while Algorithm~$\mathcal B$ returns $\Acc$. The sought probability conditioned on $D_0 +  D_0' + D_1 + D_1' = 0$ is thus
  \begin{align}
    &\Pr \Biggl[ \frac{c'P - cP}{bc'-b'c} \label{eq:e_sec} \\
    &\quad \geq \left(\frac{c'(1-e^{-\eta\nu}) - c(1-e^{-\eta\nu'})}{bc'-b'c} - \Delta_0\right) \frac{N}{2} \Biggr], \nonumber
  \end{align}
where the conditioning event imposes $P = D_2 + D_{\geq 3}$ and $P' = D_2' + D_{\geq 3}'$. In Game~\ref{game:security}, the possibly malicious Receiver is given $(c_n)_n$, the number of $n$-photon pulses transmitted by Sender, and returns $(d_n)_n$, the number of acknowledged pulses with $n$-photons. By construction $c_n$ is the sum of two binomially distributed variables $C_n$ and $C_n'$ each being the number of $n$-photon pulses within $N/2$ pulses with intensity $\nu$ and $\nu'$ respectively. That is $C_n \sim B(N/2, \nu^n e^{-\nu} / n!)$  and $C_n' \sim B(N/2, {\nu'}^n e^{-\nu'} / n!)$ where $B$ denotes the binomial distribution. Similarly, while $d_n$ is fixed by the Receiver, it is the sum of the two random variables $D_n$ and $D_n'$ introduced in the previous proof. More precisely, $D_n$ is the number of acknowledged pulses with $n$ photons that correspond to intensity $\nu$. Since $n$-photon pulses with intensity $\nu$ are indistinguishable from those with intensity $\nu'$, we can conclude that $D_n$ follows an hypergeometric distribution $H(c_n,d_n,C_n)$. Respectively we have $D_n'\sim H(c_n,d_n,C_n')$.

The proof proceeds with the following steps. (i) Note that $D_{\geq 3} \leq C_{\leq 3}$ and $0 \leq D_{\geq 3}'$ so that the probability in Eq.~\ref{eq:e_sec} can be upper bounded by
\begin{align}
  &\Pr \Biggl[ \frac{c'D_2 - cD_2' + c'C_{\geq 3}}{bc'-b'c} \label{eq:ubound} \\
  &\quad \geq \left(\frac{c'(1-e^{-\eta\nu}) - c(1-e^{-\eta\nu'})}{bc'-b'c} - \Delta_0\right) \frac{N}{2}  \Biggr]. \nonumber
\end{align}
(ii) Use tail bounds for the hypergeometric distribution to show that $D_2/C_2$ and $D_2'/C_2'$ concentrate around $d_2/c_2$, their common average value. (iii) Use tail bounds for the binomial distribution to show that $C_2$ and $C_2'$ concentrate around $cN/2$ and $c'N/2$ respectively. (iv) Use $\frac{d_2}{c_2} \times c \frac{N}{2}$ instead of $D_2$ and $\frac{d_2}{c_2} \times c' \frac{N}{2}$ instead of $D_2'$ in Eq.~\ref{eq:ubound} to show that when $D_2, D_2', C_2, C_2'$ are close to their expected values, $\varepsilon_{\text{sec}}$ is upper bounded by
\begin{align}
  & \Pr \Biggl[ \frac{c'C_{\geq 3}}{bc'-b'c} \label{eq:uubound} \\
  & \geq \left(\frac{c'(1-e^{-\eta\nu}) - c(1-e^{-\eta\nu'})}{bc'-b'c} - (\Delta_0 + \Delta_0')\right)\frac{N}{2} \Biggr], \nonumber
\end{align}
for $\Delta_0$ used in Algorithm~$\mathcal B$ and any constant $\Delta_0' > 0$.
(v) Use that $C_{\geq 3}$ concentrates around $(1-a-b-c) \frac{N}{2}$ to show  that the probability in Eq.~\ref{eq:uubound} is negligible in ${\Delta_0''}^2 N$ whenever there exists $\Delta_0''$ a positive constant such that Eq.~\ref{eq:uuubound} is satisfied.

The proof is then concluded by combining the above case where the random variables are close to their expectation values, with the ones where they don't and for which the conditional probability of Algorithm~$\mathcal B$ to return $\Acc$ can be safely bounded by 1 as these events happen with probability  negligible in $N$.
\end{proof}
Our security analysis allows to refine the scaling of the number of pulses with the transmittance of the channel in the low-transmittance regime.
\begin{theorem}[Scaling for $\eta \rightarrow 0$]
  \label{thm:scaling}
  \newcounter{count:scaling}
  \setcounterref{count:scaling}{thm:scaling}
  For two pulses of intensities $\nu, \nu'$  with $\nu  =  \alpha\nu'$ and $0<\alpha<1$, the number of pulses needed to obtain a given correctness and security error scales as $1/\eta^{2}$ for $\eta \rightarrow 0$.
  \end{theorem}
The detailed proof of Theorem~\ref{thm:scaling} is given in Appendix~\ref{sec:app-scaling}.

\begin{proof}[Scaling --- Proof Sketch]
  To obtain the scaling of $N$ with respect to $\eta$, one needs to examine all the conditions for obtaining correctness and security as $\eta \rightarrow 0$. The goal is to first characterize how $\nu$ and $\nu'$, the two pulse intensities, so that these conditions are met, and then to evaluate how this impacts the bounds on the correctness and security errors.

  The full proof given in Appendix~\ref{sec:app-scaling} shows that for small enough $\eta$, $\nu$ and $\nu'$ can made constant, and that the bounds on the correctness and security error are of the form $\exp(-\kappa \eta^2 N)$ for some $\kappa$ independent of $\eta$ and $N$. This then implies that AC-security in the low $\eta$ regime can be obtained for $N \propto 1/\eta^2$.
\end{proof}

\paragraph{Maximal reachable intensity.}
Beyond the advantageous scaling of our protocol at low transmittance, another figure of merit for such protocols is the highest intensity they can accommodate while remaining secure.

We consider a regime where the Receiver is malicious and all quantities have converged to their expected values. In this situation, the only quantity a malicious Receiver can use to reproduce the expected statistics is $c' D_{\geq 3} -c D'_{\geq 3}$ which is upper bounded by $c'C_{\geq 3}$.
Hence a stricter security condition based on this upper bound is
\begin{align}
  \frac{N}{2}\;\frac{c'(1-e^{-\eta_0 \nu}) - c\,(1-e^{-\eta_0 \nu'})}{b c'-b' c} > \frac{c'C_{\geq 3}}{b c'-b' c}
  \label{eq:GLMO_malicious}
\end{align}
where $\eta_0$ denotes the Sender's lower bound on the channel transmittance. Parameterising the two intensities as $\nu = \alpha \nu'$ with $0 < \alpha < 1$ as in Theorem~\ref{thm:scaling} and solving the equality case of Eq.~(\ref{eq:GLMO_malicious}), yields the
maximal intensity value $\nu^*_{\rm GLMO}$ of $\nu'$ at which security breaks down.

The same reasoning can be applied to the single-pulse protocol of \cite{DKL11universal}. In the same setting, the security criterion reads
\begin{align}
  N\,\left(1-e^{-\eta_0 \nu}\right) > C_{\geq 2}
  \label{eq:DKL_malicious}
\end{align}
and we call $\nu^*_{\rm DKL}$ the value of $\nu$ at which Eq.~(\ref{eq:DKL_malicious}) fails.

In the following, we compare the two values of the maximal reachable intensity. We numerically compute $\nu^*_{\rm GLMO}$ while $\nu^*_{\rm DKL}$ can be found analytically. It reads
\begin{align}
  \nu^*_{\rm DKL} = \frac{W_{-1}\left((\eta_0 - 1)e^{\eta_0 - 1}\right)}{\eta_0 - 1} - 1
\end{align}
where $W_{-1}$ is the negative branch of the Lambert function.

Figure~\ref{fig:fig_eta} shows $\nu^*_{\rm DKL}$ and $\nu^*_{\rm GLMO}$ as a function of $\eta_0$ for fixed $\alpha = 0.5$. It shows that at low transmittance, our protocol accommodates higher values of intensity.
\begin{figure}[ht]
  \centering
  \includegraphics[width=0.7\linewidth]{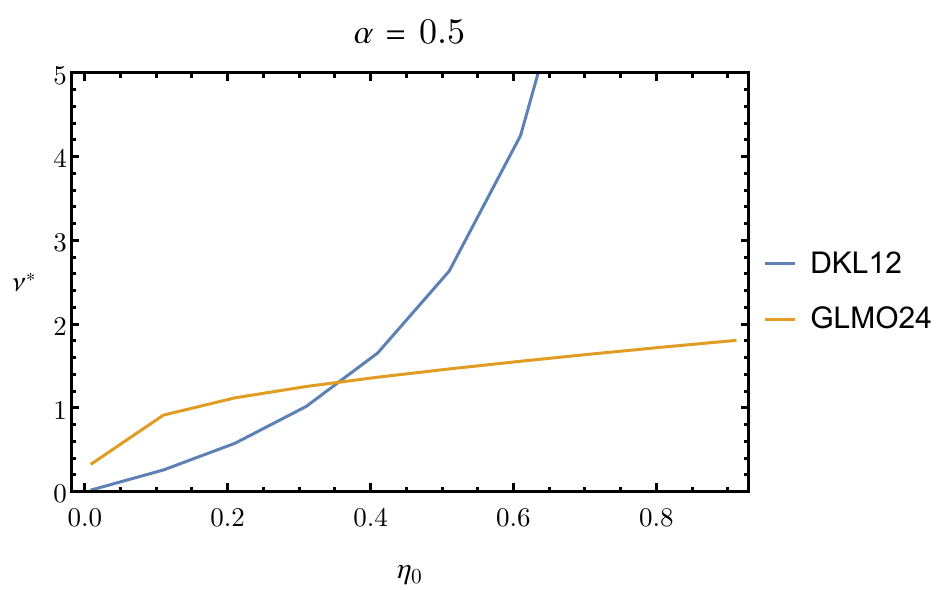}
  \caption{Maximal intensity value $\nu^*$ allowed by both protocols as a function of $\eta_0$ for fixed $\alpha = 0.5$. The blue curve represents the results for the protocol of \cite{DKL11universal} and the orange one the results for our protocol. This plot shows that at low transmittance, our protocol remains secure for higher intensities values.}
  \label{fig:fig_eta}
\end{figure}
Figure~\ref{fig:fig_alpha} shows the same quantities for fixed $\eta_0 = 0.2$ and varying $\alpha$.
\begin{figure}[ht]
  \centering
  \includegraphics[width=0.7\linewidth]{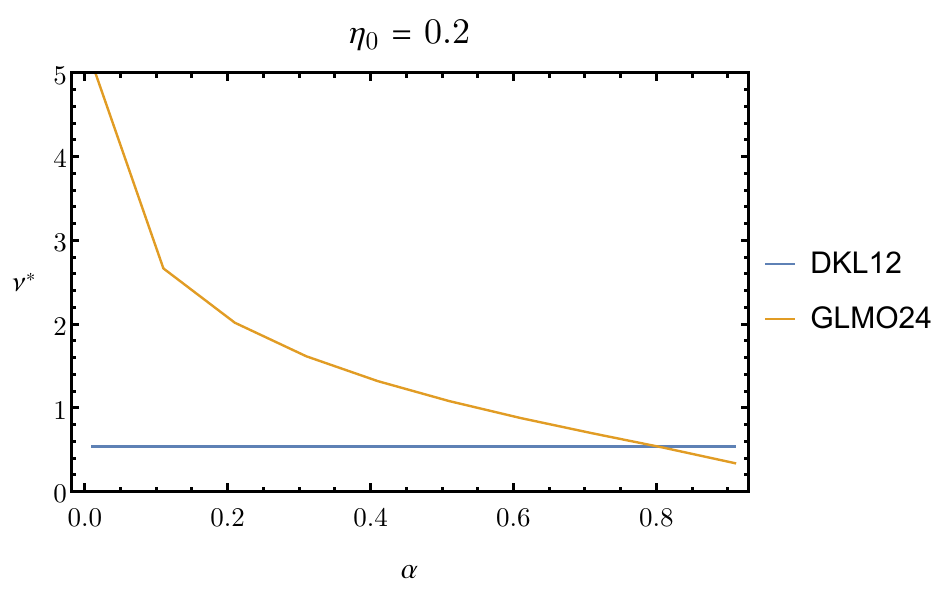}
  \caption{Maximal intensity value $\nu^*$ allowed by both protocols as a function of $\alpha$ for fixed $\eta_0= 0.2$. The blue curve represents the results for the protocol of \cite{DKL11universal} and the orange one the results for our protocol. This plot shows that at low transmittance, our protocol remains secure for higher intensities values.}
  \label{fig:fig_alpha}
\end{figure}
Figure~\ref{fig:density_plot} shows which protocol has the highest $\eta^*$ value as a function of both $\alpha$ and $\eta_0$.
\begin{figure}[ht]
  \centering
  \includegraphics[width=0.7\linewidth]{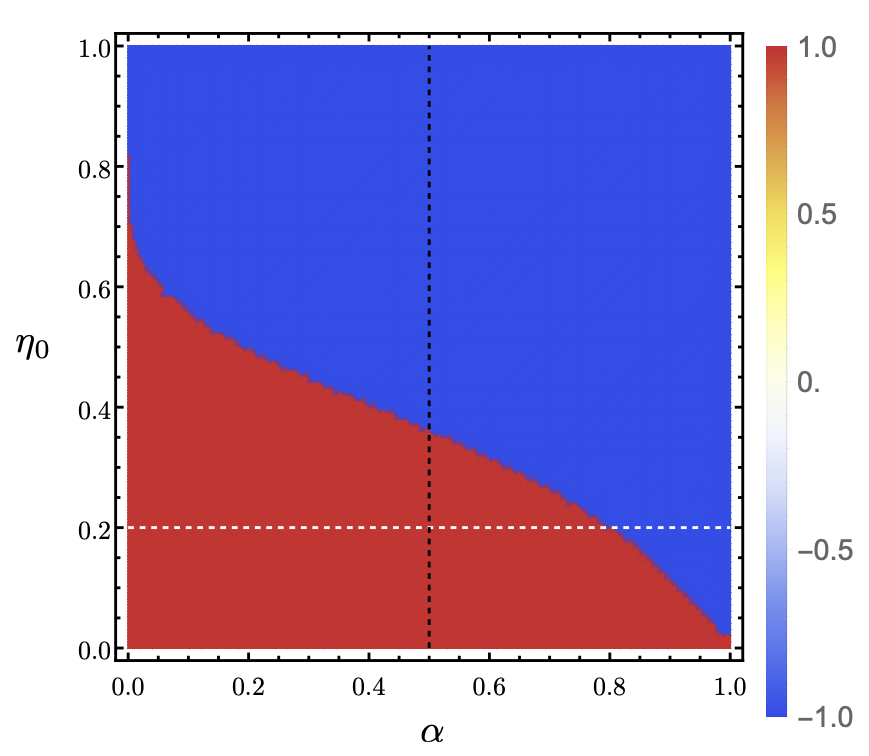}
  \caption{In the $(\alpha,\eta_0)$ plane, the colour shows which protocol accommodates the highest $\nu^*$ value: red for our protocol and blue for the \cite{DKL11universal} one. The black (resp. white) dashed line represent the line along which Fig.~\ref{fig:fig_eta} (resp. Fig.~\ref{fig:fig_alpha}) is plotted.}
  \label{fig:density_plot}
\end{figure}
These results show that depending on the estimated channel transmittance and the available resources (laser power, number of pulses), one can choose either of the protocols.
The precise exploration of all parameter trade-offs, the optimisation of the protocol and the its practical implementation are left for future work.

\section{Discussion}
\label{sec:discussion}
\paragraph{Intuition for the accept criterion of Algorithm $\mathcal B$}
To obtain an intuition behind the choice made in Equation~\ref{eq:T}, we can simplify the problem by considering that there are no statistical finite-size fluctuations. We next consider the optimal attack by a malicious receiver: the replacement of the lossy fiber by a perfect one, the rejection of all single-photon pulses, which leaves the two-and-more photon pulses to satisfy the acceptance criteria set by the protocol. Next we can examine what his attack could look like for pulses with exactly two photons. Because it does not know whether a given two-photon pulse corresponds to intensity $\nu$  or $\nu'$, the most it can do  will be  to set a ratio $\beta$ of two-photon pulses that will be acknowledged as received, the rest being declared as not received. Hence, the number of acknowledged 2-photon pulses emitted with intensity $\nu$ will be $\beta c \frac{N}{2}$, while those emitted with intensity $\nu'$ will be $\beta c' \frac{N}{2}$.

We can now remember the definition of the observed quantity in Equation~\ref{eq:T} that is computed by the sender:
\begin{equation}
  T = \frac{1}{bc'-b'c}\left(c'P - cP' \right),
\end{equation}
so that the contribution of 2-photon pulses reported by the receiver is equal to $\beta(c c' - c' c)/(bc' - b'c) = 0$.  In other words, the choice of $T$ is crafted in such a way that, whatever the strategy of the attacker, the contribution from the most numerous events to reaching $t$, the honest but lossy value of $T$, is always 0. The malicious reciever will thus have to reach $t$ with 3-and-more photon pulses, which is harder to achieve and thus explains the better scaling behavior for low $\eta$ of our protocol when compared to~\cite{DKL11universal}.

The proof in Appendix~\ref{proof:ac-security} makes these arguments rigorous by taking into account finite-size effects and by showing that in spite of the coherent nature of the attack, one can still assign a ratio $\beta$ for the 2-photon behavior up to statistical fluctuations, that again need to be properly accounted for.

The above argument also indicates how to generalize the protocol to more pulses and achieve improved scaling. We leave this generalization for future work.

\paragraph{Improved Security}
We close possible gaps in standalone security proofs and provide solid foundations for other WCP-based protocols by showing the composable security of our protocol. Furthermore our work doesn't assume that the source perfectly follows the Poisson distribution, an assumption that is made by some other composable security proofs.
The composable security ofthe $\mathsf{BatchRSP}$ resource implies that it can be combined with different privacy amplification techniques (like the classical one used for QKD) without requiring a new proof.

\paragraph{Improved practicality}
As far as implementability is concerned, using decoy-state-based methods yields an immediate advantage in terms of achievable RSP rates.  Indeed, not using decoys requires to send a number of pulses that grows as $1/\eta^4$ in order to send a single secure state~\cite{DKL11universal}, while decoy states allow to reduce the scaling to $1/\eta$~\cite{JWHX19remote}. Our technique yields an intermediate scaling of $1/\eta^2$ when two different intensities are used with additional security and verifiability benefits.  However, a crucial difference between our method and the traditional decoy state one is that we use all single-photon states to transmit quantum information and not only from a single intensity. Moreover, an important property of our protocol is that it can be streamed, i.e.~the server does not need to wait for information from the client before acting on the received pulses. Indeed, that is th case due to the fact that the estimation procedure only depends on the photon number in the pulse and not on the actual single-qubit state that is encoded into the pulse. Finally, the implementation of our protocol doesn't require additional resources compared to \cite{DKL11universal,JWHX19remote}. We believe these properties pave the way for the experimental demonstration of longer delegated quantum computations

\paragraph{Improved security for QKD}
Standard QKD and RSP protocols share the same vulnerability if the single-photon sources are replaced by weak coherent pulses. Multiple photon pulses could in principle be leveraged by an attacker to compromise the secrecy of the scheme and recover either some bits of the key (for QKD) or the value of the quantum One-Time-Pad used to blind the computation (for RSP). As a matter of fact a successful photon-number-splitting attack works similarly in both cases. For the former, it consists of (i) replacing the noisy fiber between the sender, the eavesdropper and the receiver, and (ii) have the eavesdropper intercept all single-photon pulses while trying to mimick the lossy fiber with the rest of the pulses, while storing extraneous photons until the reconciliation phase. In the latter, it consists of (i) replacing the noisy fiber between the client and server's location with a noiseless one, and (ii) acknowledging a successful transmission only whenever the received pulse contains two or more photons.

The applicability of our security proof to both cases results from (i) the ability to capture the impactful imperfections of WCP sources through the same model (Section~\ref{sec:wcpg}), (ii) the similarity  of the optimal attacks and (iii) the composability of the proof which allows to choose either a classical or quantum secure privacy amplification.

As was examplified in Section~\ref{sec:example}, to concretely prove the security of the scheme, one has to define an accept/reject criterion and show that the accept case can only be passed with negligible probability while performing a successful attack. We showed in Section~\ref{sec:batch-rsp} that the most general coherent attack amount to defined for each $n \in \mathbb N$, how many of the $n$-photon pulses emitted by the sender will be acknowledged by the receiver.  While it is expected that for several intensities the accepted pulses for a given $n$  will be spread throughout the various intensities according to the relative frequencies of their $n$ photon states, it cannot be positively claimed until the empirical averages have converged toward their expectation values. This is a simplification that is, to the best of our knowledge, usually made when proving the security of decoy-sate QKD, but which cannot be justified especially for high $n$. In contrast, Section~\ref{sec:example} does not make this assumption. This comes at the price of a lengthier proof and the worse scaling $1/\eta^2$ instead of $1/\eta$. Yet, a positive consequence of our proof technique is that it is robust to cahnges in the photon number distribution of the pulses (cf. Remark~\ref{rem:robustness}).

\paragraph{Acknowledgements}
The authors would like to thank P. Hilaire, U. Chabaud and N. Fabre for useful discussions. DL acknowledges support by the Engineering and Physical Sciences Research Council [grant reference EP/T001062/1], and the UK National Quantum Computer Centre [NQCC200921], which is a UKRI Centre and part of the UK National Quantum Technologies Programme (NQTP). This work has been co-funded by the European Commission as part of the EIC accelerator program under the grant agreement 190188855 for SEPOQC project, by the Horizon-CL4 program under the grant agreement 101135288 for EPIQUE project, by the  Horizon-CL4 program under Grant agreement 101102140 for QIA-Phase 1 project, by ANR research grant ANR-21-CE47-0014 (SecNISQ), and by HQI Initiative supported by France 2030 under ANR grant ANR-22-PNCQ-0002.

\bibliographystyle{alpha}
\bibliography{../qubib/qubib.bib}

\newcommand{\etalchar}[1]{$^{#1}$}
\begin{thebibliography}{DNM{\etalchar{+}}24}

\bibitem[Aar07]{A07scott}
Scott Aaronson.
\newblock {T}he {S}cott {A}aronson 25.00\$ {P}rize.
\newblock \url{http://www.scottaaronson.com/blog/?p=284}, October 2007.
\newblock Accessed: Jan. 30 2015.

\bibitem[ABOE10]{ABE10interactive}
Dorit Aharonov, Michael Ben-Or, and Elad Eban.
\newblock Interactive proofs for quantum computation.
\newblock In {\em Proceedings of Innovations of Computer Science (ICS 2010)},
  page 453–469, 2010.

\bibitem[AV13]{AV13is}
Dorit Aharonov and Umesh~V. Vazirani.
\newblock {\em Is Quantum Mechanics Falsifiable? A Computational Perspective on
  the Foundations of Quantum Mechanics}, page 329–350.
\newblock The MIT Press, June 2013.

\bibitem[BFK09]{BFK09universal}
Anne Broadbent, Joseph Fitzsimons, and Elham Kashefi.
\newblock Universal blind quantum computation.
\newblock In IEEE, editor, {\em 50th Annual IEEE Symposium on Foundations of
  Computer Science}, 2009.

\bibitem[Cas23]{C23ibm}
Davide Castelvecchi.
\newblock Ibm releases first-ever 1,000-qubit quantum chip.
\newblock {\em Nature}, 624(7991):238–238, December 2023.

\bibitem[DKL12]{DKL11universal}
Vedran Dunjko, Elham Kashefi, and Anthony Leverrier.
\newblock Universal blind quantum computing with weak coherent pulses.
\newblock {\em Phys. Rev. Lett.}, 108(200502), 2012.

\bibitem[DKP07]{DKP07measurement-calculus}
Vincent Danos, Elham Kashefi, and Prakash Panangaden.
\newblock The measurement calculus.
\newblock {\em J. ACM}, 54(2), April 2007.

\bibitem[DNM{\etalchar{+}}24]{DNMN23verifiable}
P.~Drmota, D.~P. Nadlinger, D.~Main, B.~C. Nichol, E.~M. Ainley, D.~Leichtle,
  A.~Mantri, E.~Kashefi, R.~Srinivas, G.~Araneda, C.~J. Ballance, and D.~M.
  Lucas.
\newblock Verifiable blind quantum computing with trapped ions and single
  photons.
\newblock {\em Phys. Rev. Lett.}, 132:150604, Apr 2024.

\bibitem[FK17]{FK17unconditionally}
Joseph~F. Fitzsimons and Elham Kashefi.
\newblock Unconditionally verifiable blind quantum computation.
\newblock {\em Phys. Rev. A}, 96:012303, Jul 2017.

\bibitem[GLMO24]{GLMO24composably}
Maxime Garnier, Dominik Leichtle, Luka Music, and Harold Ollivier.
\newblock Composably secure delegated quantum computation with weak coherent
  pulses.
\newblock In {\em 2024 International Conference on Quantum Communications,
  Networking, and Computing (QCNC)}, page 221–225. IEEE, July 2024.

\bibitem[Got04]{G04conference}
Daniel Gottesman.
\newblock Conference as reported in~\cite{ABE10interactive}, 2004.

\bibitem[JWH{\etalchar{+}}19]{JWHX19remote}
Yang-Fan Jiang, Kejin Wei, Liang Huang, Ke~Xu, Qi-Chao Sun, Yu-Zhe Zhang,
  Weijun Zhang, Hao Li, Lixing You, Zhen Wang, Hoi-Kwong Lo, Feihu Xu, Qiang
  Zhang, and Jian-Wei Pan.
\newblock Remote blind state preparation with weak coherent pulses in the
  field.
\newblock {\em Physical Review Letters}, 123(10), September 2019.

\bibitem[KKL{\etalchar{+}}24]{KKLM22unifying}
Theodoros Kapourniotis, Elham Kashefi, Dominik Leichtle, Luka Music, and Harold
  Ollivier.
\newblock Unifying quantum verification and error-detection: Theory and tools
  for optimisations.
\newblock {\em Quantum Science and Technology}, 9(3):035036, May 2024.

\bibitem[KKL{\etalchar{+}}25]{KKLM23asymmetric}
Theodoros Kapourniotis, Elham Kashefi, Dominik Leichtle, Luka Music, and Harold
  Ollivier.
\newblock Asymmetric quantum secure multi-party computation with weak clients
  against dishonest majority.
\newblock {\em Quantum Science and Technology}, 10(2):025015, 2025.

\bibitem[KLMO24]{KLMO24verification}
Elham Kashefi, Dominik Leichtle, Luka Music, and Harold Ollivier.
\newblock Verification of quantum computations without trusted preparations or
  measurements.
\newblock {\em CoRR}, 2024.

\bibitem[LCW{\etalchar{+}}14]{LCWX14concise}
Charles Ci~Wen Lim, Marcos Curty, Nino Walenta, Feihu Xu, and Hugo Zbinden.
\newblock Concise security bounds for practical decoy-state quantum key
  distribution.
\newblock {\em Physical Review A}, 89(2), February 2014.

\bibitem[LMC05]{LMC05decoy}
Hoi-Kwong Lo, Xiongfeng Ma, and Kai Chen.
\newblock Decoy state quantum key distribution.
\newblock {\em Physical Review Letters}, 94(23), June 2005.

\bibitem[LMKO21]{LMKO21verifying}
Dominik Leichtle, Luka Music, Elham Kashefi, and Harold Ollivier.
\newblock Verifying {BQP} computations on noisy devices with minimal overhead.
\newblock {\em PRX Quantum}, 2:040302, Oct 2021.

\bibitem[Mah18]{M18classical}
Urmila Mahadev.
\newblock Classical verification of quantum computations.
\newblock In Mikkel Thorup, editor, {\em 59th {IEEE} Annual Symposium on
  Foundations of Computer Science, {FOCS} 2018, Paris, France, October 7-9,
  2018}, pages 259--267. {IEEE} Computer Society, 2018.

\bibitem[MR11]{MR11abstract-cryptography}
Ueli Maurer and Renato Renner.
\newblock Abstract cryptography.
\newblock In {\em Innovations in Computer Science}, pages 1 -- 21. Tsinghua
  University Press, jan 2011.

\bibitem[TM24]{TM24unconditional}
Yuki Takeuchi and Akihiro Mizutani.
\newblock Unconditional verification of quantum computation with classical
  light.
\newblock {\em CoRR}, 2024.

\end{thebibliography}

\newpage
\appendix

\section{Proof of Theorem~\ref{thm:decoy-realization}}
\label{app:proof-thm-protocol}

We prove here the result from Section~\ref{sec:batch-rsp}, restated below.
\newcounter{temp-count}
\setcounter{temp-count}{\value{theorem}}
\setcounter{theorem}{\value{count:decoy-realization}-1}
\begin{theorem}
  For given transmittance $\eta \in [0,1]$, parameters $N,K \in \mathbb{N}^+$, $\mu_1, \ldots, \mu_N \in \mathbb{R}^+_0$, and algorithm $\mathcal{B}$, Protocol~\ref{proto:decoy-state-method} $\varepsilon$-statistically constructs Resource~\ref{res:batchrsp} ($\mathsf{BatchRSP}$) from Resource~\ref{res:wcp} ($\mathsf{WCPGenerator}$) if Game~\ref{game:correctness} returns $\Abort$ with probability at most $\varepsilon$, and Game~\ref{game:security} returns $\Fail$ with probability at most $\varepsilon$ for all maps $\mathcal{D}$.
\end{theorem}
\setcounter{theorem}{\value{temp-count}}

We prove the correctness and security separately.

\begin{lemma}[Correctness of the Decoy State Method]\label{lemma:decoy-correctness}
  For given transmittance $\eta \in [0,1]$, parameters $N,K \in \mathbb{N}^+$, $\mu_1, \dots, \mu_N \in \mathbb{R}^+_0$, and algorithm $\mathcal{B}$, Protocol~\ref{proto:decoy-state-method} $\varepsilon$-correctly realises Resource~\ref{res:batchrsp} ($\mathsf{BatchRSP}$) from Resource~\ref{res:wcp} ($\mathsf{WCPGenerator}$) if Game~\ref{game:correctness} returns $\Abort$ with probability at most $\varepsilon$.
\end{lemma}

\begin{proof}[Proof of Correctness]
  We show that the output state of Protocol~\ref{proto:decoy-state-method} with honest Sender and Receiver is the same as the one produced by the Batch RSP (Resource \ref{res:batchrsp}). 
  
  The probability that Game~\ref{game:correctness} returns $\Abort$ is exactly equal to the probability that the Sender or Receiver abort Protocol~\ref{proto:decoy-state-method}.
  
  We then analyse the behaviour of Protocol~\ref{proto:decoy-state-method} conditioned on the event that it does not abort. The states held by the Receiver after the acknowledgement of the WCP reception are of the form $( \cptp U_i' (\rho_0) )_{i\in I'}$. Assuming that the Sender's decoy estimation subroutine accepts, the corrections by the Receiver are of the form $( \cptp U_{\sigma(j)} \cptp U_j'^\dagger )_{j = 1,\dots,K}$ after the relabelling. After the application of these corrections, the Receiver's output takes the form $( \cptp U_{\sigma(j)} (\rho_0) )_{j = 1,\dots,K}$. After permuting these states by $\sigma$, this yields exactly the same output as that of Resource~\ref{res:batchrsp} ($\mathsf{BatchRSP}$).
  
  Therefore the correctness error is exactly equal to the abort probability in Protocol~\ref{proto:decoy-state-method},~i.e. the probability that Game~\ref{game:correctness} returns $\Abort$.

\end{proof}

\begin{lemma}[Security of the Decoy State Method]\label{lemma:decoy-security}
	For given parameters $N,K \in \mathbb{N}^+$, $\mu_1, \dots, \mu_N \in \mathbb{R}^+_0$, and algorithm $\mathcal{B}$, Protocol~\ref{proto:decoy-state-method} $\varepsilon$-securely realises Resource~\ref{res:batchrsp} ($\mathsf{BatchRSP}$) from Resource~\ref{res:wcp} ($\mathsf{WCPGenerator}$) against a malicious Receiver if Game~\ref{game:security} returns $\Fail$ with probability at most $\varepsilon$ for all maps $\mathcal{D}$.
\end{lemma}

\begin{proof}[Proof of Soundness]
  To show the security of Protocol~\ref{proto:decoy-state-method} against a malicious Receiver, we make use of the explicit construction of Simulator~\ref{sim:decoy-state-method}. We then prove that no Distinguisher can tell apart the malicious Receiver's interactions in (i) the ideal world in which the Simulator has single-query oracle access to the $\mathsf{BatchRSP}$ Resource \ref{res:batchrsp}, and (ii) the real world in which the honest Sender interacts with the $\mathsf{WCPGenerator}$ Resource~\ref{res:wcp}.
  
  The simulator informally works as follows. It first emulates the calls to the $\mathsf{WCPGenerator}$ Resource~\ref{res:wcp}. It samples a photon number $n_i$ from the Poisson distribution associated to the pulse intensity $\mu_i$ for all $N$ pulses -- permuted as in the protocol. If $n_i = 0$, it sends the vacuum state to the Receiver, if $n_i = 1$ it sends half of an EPR pair and keeps the other half, and otherwise if $n_i > 1$ it sends $n_i$ and the description a unitary $\cptp U_i'$ sampled from the Haar measure on $\mathcal{U}$. After receiving a set $I$ of accepted pulses of size $K$ (otherwise the Receiver aborts), it performs the same check as an honest Sender by verifying that algorithm $\mathcal{B}$ returns $\Acc$ on this set. At this step, if there are no indices $i \in I$ such that $n_i \leq 1$, the Simulator returns $\Fail$.\footnote{In the real protocol this corresponds to the case where the adversary is able to obtain the classical description of all the pulses they accepted, meaning that they have complete knowledge about the final state. The simulator is incapable of reproducing this behaviour.} Otherwise, it then performs the oracle call to the $\mathsf{BatchRSP}$ resource and receives and index $k \leq K$, classical description of the unitaries $\{\cptp U_i\}_{i\neq k}$, and a quantum state $\rho_k$. It chooses at random an index $i^\ast \in [1,K]$ (after relabelling) such that $n_{i^\ast} \leq 1$ and teleports the state received from the $\mathsf{BatchRSP}$ resource if $n_{i^\ast} = 1$. The correction permutation is chosen such that this state appears at index $k$ in the Receiver's system after applying it. The rest of the kept states can be perfectly emulated using the classical description of unitaries $\{\cptp U_i\}_{i\neq k}$.
  
  \begin{figure}[htp]
  \begin{simulator}[H]
    \caption{Against Malicious Receiver}
    \label{sim:decoy-state-method}
    \begin{algorithmic} 
      
      \Procedure{\textbf{Emulation of} $N$ \textbf{Calls to} $\mathsf{WCPGenerator}$}{}
      \State Sample a random permutation $\pi \sample \operatorname{S}_N$.
    	  \State Apply $\pi$ to permute the elements of the tuple $(\mu_1, \dots, \mu_N)$ to obtain $(\mu_1', \dots, \mu_N')$.
   	  \For{$i \in [1,N]$}
    	  \State Sample $\cptp U_i' \sample \mathcal{U}$ from the Haar measure on $\mathcal{U}$.
      \State Sample $n_i \sample \Pois(\mu_i')$.
      \If{$n_i = 0$}
      \State Send vacuum to the Receiver.
      \ElsIf{$n_i = 1$}
      \State Send half of an EPR-pair to the Receiver and keep the other half.
      \Else %
      \State Send $(n_i, \cptp U_i')$ to the Receiver.
      \EndIf
      \EndFor
      \EndProcedure

      \Procedure{\textbf{Emulation of the Decoy Estimation}}{}
      \State Receive $I'$ or $\Abort$ from the Receiver.
      \If{It receives $\Abort$}
      \State Send $c=1$ and $d=1$ to the $\mathsf{BatchRSP}$ resource and stop.
      \EndIf
      \State Undo the permutation, \textit{i.e.} let $I \gets \{ \pi^{-1}(i) | i \in I' \}$.
      \State Perform the decoy estimation routine by running $\mathcal{B}$ on input $I$.
      \If{$\mathcal{B}$ does not return $\Acc$}
      \State Send $c=1$ and $d=1$ to the $\mathsf{BatchRSP}$ resource and stop.
      \EndIf
      \State Relabel the kept state indexed previously by $I'$ from $1$ to $K$.
      \If{there exists no $j^\ast \in [1, K]$ such that $n_{j^\ast} \leq 1$}
      \State Return $\Fail$ and stop.
      \EndIf
      \State Send $c=1$ and $d=0$ to the $\mathsf{BatchRSP}$ resource. Receive in return index $k$, unitaries $\{\cptp U_j\}_{j\neq k}$, and a quantum state $\rho_k$.
      \State Let $j^\ast \in [1, K]$ be a random index such that $n_{j^\ast} \leq 1$.
      \State Sample a random permutation $\sigma \sample \operatorname{S}_K$ such that $\sigma(j^\ast) = k$.
      \State Apply $\cptp U_{j^\ast}'$ to $\rho_k$ to obtain $\rho_k' = \cptp U_{j^\ast}'(\rho_k)$.
      \For{$j \in [1,K]$}
      \If{$j = j^\ast$ and $n_{j^\ast} = 0$}
      \State Let $\tilde{\cptp U}_{j^\ast} = \cptp U \cptp U_{j^\ast}'^\dagger$ for a random $\cptp U$.
      \ElsIf{$j = j^\ast$ and $n_{j^\ast} = 1$}
      \State Perform a Bell measurement of the half EPR pair at position $j^\ast$ and the quantum state $\rho_k'$. Denote the outcomes of the Bell measurement by $x_{j^\ast}$ and $z_{j^\ast}$.
      \State Let $\tilde{\cptp U}_{j^\ast} = \cptp U_{j^\ast}'^\dagger \X^{x_{j^\ast}} \Z^{z_{j^\ast}}$.
      \ElsIf{$n_j = 1$}
      \State Prepare the quantum state $\cptp U_j' \cptp U_{\sigma(j)} (\rho_0)$.
      \State Perform a Bell measurement of the half EPR pair at position $j$ and this quantum state. Denote the corrections of the Bell measurement by $x_j$ and $z_j$.
      \State Let $\tilde{\cptp U}_j = \cptp U_j'^\dagger \X^{x_j} \Z^{z_j}$.
      \Else
      \State Let $\tilde{\cptp U}_j = \cptp U_{\sigma(j)} \cptp U_j'^\dagger$.
      \EndIf
      \EndFor
      \State Send the corrections $\sigma$, $(\tilde{\cptp U}_j)_{j\in [1, K]}$ to the Receiver.
      \EndProcedure
    \end{algorithmic}
  \end{simulator}
  \end{figure}
  
	In the following we distinguish between the two events denoted by $\Abort$ and $\Fail$. While $\Abort$ events indicate that the protocol rejects and hence can occur also in the real world, $\Fail$ is raised only by the simulator in the ideal world and indicates a failure of the simulator in emulating real world behaviour. These $\Fail$ events capture exactly the cases in which the Receiver successfully cheats in the real world -- all the checks pass and only pulses with two or more photons are kept -- and will represent the only difference between the ideal and the real world, and any distinguishing advantage will stem entirely from these events. Consequently, it will be a core element of this security proof to upper-bound the probability of such $\Fail$ event.
	
	The proof is split into two steps. First, we show that for any distinguisher of the real and ideal worlds, there exists a map $\mathcal{D}$ such that the probability of Simulator~\ref{sim:decoy-state-method} returning $\Fail$ equals the probability of \textsc{GameSim}($\mathcal{D}$) returning $\Fail$. The latter, in turn, is upper-bounded by $\varepsilon$ by the condition of Lemma~\ref{lemma:decoy-security}. Secondly, we show that the interaction with the simulator in the ideal world is perfectly indistinguishable from the interaction with the Sender in the real world, conditioned on the event that the simulator does not return $\Fail$. The final statement of the Lemma follows by application of the triangle inequality.
	
	\paragraph{Bounding simulation failure by Game~\ref{game:security}.}
	Regarding the first part of the proof, we will trim down the simulation to slowly exhibit games which have the same failure probability but are each closer to \textsc{GameSim}($\mathcal{D}$).
	
	We start by noting that the decision of whether or not to return $\Fail$ is made by the simulator before the $\mathsf{BatchRSP}$ resource is invoked, and is therefore independent of the target states and their descriptions. Consequently, a hybrid game with the same failure probability as the simulator can be obtained by simply removing the call to the $\mathsf{BatchRSP}$ resource and the corrections. In this new game, given the photon numbers, the Receiver (i.e. the distinguisher) can sample itself the unitaries and prepare the quantum states during the emulation of the $\mathsf{WCPGenerator}$, since these no longer depend on any hidden information. However, sampling the photon numbers must still be kept separate from the distinguisher since it depends on the intensities $\mu_i$, which must remain hidden from the distinguisher.
	
	In this way, any distinguisher interacting with the outer interfaces of Simulator~\ref{sim:decoy-state-method} and the ideal $\mathsf{BatchRSP}$ resource can be transformed into a quantum algorithm interacting with Game~\ref{game:security-hybrid}, with the same failure probability. 	As this quantum algorithm receives only classical inputs and returns only classical results, its influence on the failure probability of Game~\ref{game:security-hybrid} can be captured fully by a mathematical map $\mathcal{C}$ of $N$ integers to probabilistic distributions over the set of $K$-element subsets of $[1,N]$.
	
\begin{figure}[htp]
\begin{game}[H]
  \caption{Hybrid Security Game}
  \label{game:security-hybrid}
  \begin{algorithmic}[0]
    \State \textbf{Parameters:}
    \State $N,K \in \mathbb{N}^+$, where $K \leq N$.
    \State $\mu_1, \dots, \mu_N \in \mathbb{R}^+_0$.
    \State Algorithm $\mathcal{B}$.
    \Procedure{GameSimHybrid}{$\mathcal{C}$}
	\State Sample a random permutation $\pi \sample \operatorname{S}_N$.
    	\State Apply $\pi$ to permute the elements of the tuple $(\mu_1, \dots, \mu_N)$ to obtain $(\mu_1', \dots, \mu_N')$.
	\For{$i \in [1,N]$}
    \State Sample $n_i \sample \Pois(\mu_i')$.
    \EndFor
    \State Let $I' \gets \mathcal{C}(n_1, \dots, n_N)$ with $|I'| = K$.
    \State Let $I \gets \{ \pi^{-1}(i) | i \in I' \}$.
    \If{$\mathcal{B}(I)$ returns $\Acc$, and $\not\exists i^\ast \in I' : n_{i^\ast} \leq 1$}
    \State Return $\Fail$.
    \Else
    \State Return $\Succ$.
    \EndIf
    \EndProcedure
  \end{algorithmic}
\end{game}
\end{figure}
	
	To show that the failure probability of Game~\ref{game:security} equals that of Game~\ref{game:security-hybrid}, it suffices to see that the only information that is accessible to $\mathcal{C}$ in Game~\ref{game:security-hybrid} are the frequencies of the different photon counts, denoted by $(c_n)_{n \in \mathbb{N}}$ in Game~\ref{game:security}. This is the case because the positions of the different counts $n_i$ are randomised before passing them to $\mathcal{C}$ (and derandomised directly afterwards).
	Similarly, the only information in the output $I'$ of $\mathcal{C}$ that is not destroyed by the random permutation immediately applied elementwise to $I'$ is the number of accepted $n$-photon states, for all $n \in \mathbb{N}$, denoted by $(d_n)_{n \in \mathbb{N}}$ in Game~\ref{game:security}.
	Consequently, the map $\mathcal{C}$ from Game~\ref{game:security-hybrid} can be replaced with a more restricted map $\mathcal{D}$ which maps the frequencies of photon counts to the number of accepted states by photon count. This yields Game~\ref{game:security}, and concludes the first part of the proof.

	\paragraph{Perfect simulation in absence of failure.}
	In the following, we condition on the event that Simulator~\ref{sim:decoy-state-method} does not return $\Fail$. In this case, merging the Simulator with the $\mathsf{BatchRSP}$ resource results in Game~\ref{game:real-ideal-hybrid1}. %

  \begin{figure}[htp]
  \begin{game}[H]
    \caption{Hybrid Game Without $\Fail$}
    \label{game:real-ideal-hybrid1}
    \begin{algorithmic} 
    
      \Procedure{\textbf{Emulation of the $\mathsf{BatchRSP}$ resource}}{}
      \State Receive $K$ unitaries $\cptp U_1, \dots, \cptp U_K \in \mathcal{U}$.
      \EndProcedure
      
      \Procedure{\textbf{Emulation of} $N$ \textbf{Calls to} $\mathsf{WCPGenerator}$}{}
      \State Sample a random permutation $\pi \sample \operatorname{S}_N$.
    	  \State Apply $\pi$ to permute the elements of the tuple $(\mu_1, \dots, \mu_N)$ to obtain $(\mu_1', \dots, \mu_N')$.
   	  \For{$i \in [1,N]$}
    	  \State Sample $\cptp U_i' \sample \mathcal{U}$ from the Haar measure on $\mathcal{U}$.
      \State Sample $n_i \sample \Pois(\mu_i')$.
      \If{$n_i = 0$}
      \State Send vacuum to the Receiver.
      \ElsIf{$n_i = 1$}
      \State Send half of an EPR-pair to the Receiver and keep the other half.
      \Else %
      \State Send $(n_i, \cptp U_i')$ to the Receiver.
      \EndIf
      \EndFor
      \EndProcedure

      \Procedure{\textbf{Emulation of the Decoy Estimation}}{}
      \State Receive $I'$ from the Receiver of size $K$.
      \State Undo the permutation, i.e. let $I \gets \{ \pi^{-1}(i) | i \in I' \}$.
      \State Perform the decoy estimation routine by running $\mathcal{B}$ on input $I$.
      \If{$\mathcal{B}$ does not return $\Acc$}
      \State Send $\Abort$ to both parties.
      \EndIf
      \State Relabel the kept state from $1$ to $K$. Sample a random permutation $\sigma \sample \operatorname{S}_K$.
      \For{$j \in [1,K]$}
      \If{$n_j = 1$}
      \State Prepare the quantum state $\cptp U_j' \cptp U_{\sigma(j)} (\rho_0)$.
      \State Perform a Bell measurement of the half EPR pair at position $j$ and this quantum state. Denote the outcomes of the Bell measurement by $x_j$ and $z_j$.
      \State Let $\tilde{\cptp U}_j = \cptp U_j'^\dagger \X^{x_j} \Z^{z_j}$.
      \Else
      \State Let $\tilde{\cptp U}_j = \cptp U_{\sigma(j)} \cptp U_j'^\dagger$.
      \EndIf
      \EndFor
      \State Send the corrections $\sigma$, $(\tilde{\cptp U}_j)_{j\in [1, K]}$ to the Receiver.
      \EndProcedure
    \end{algorithmic}
  \end{game}
  \end{figure}
	
	Next, instead of using the half EPR pairs to teleport the quantum states through the single-photon pulses, they can be prepared directly in the state $\cptp U_i' \cptp U_{\sigma(i)} (\rho_0)$ from the start without changing the output distribution.
	Finally, a simple substitution of the random unitaries drawn from the Haar measure for the case $n_i = 1$ yields Game~\ref{game:real-ideal-hybrid2} which is therefore perfectly indistinguishable from Game~\ref{game:real-ideal-hybrid1}.\footnote{The corrections for $n_i = 1$ and $n_i \neq 1$ are the same now and the cases can therefore be merged.}

  \begin{figure}[hp]
  \begin{game}[H]
    \caption{Hybrid Game Without $\Fail$ - No Teleportation}
    \label{game:real-ideal-hybrid2}
    \begin{algorithmic} 
    
      \Procedure{\textbf{Emulation of the $\mathsf{BatchRSP}$ functionality}}{}
      \State Receive $K$ unitaries $\cptp U_1, \dots, \cptp U_K \in \mathcal{U}$ on its left interface.
      \EndProcedure
      
      \Procedure{\textbf{Emulation of} $N$ \textbf{Calls to} $\mathsf{WCPGenerator}$}{}
      \State Sample a random permutation $\pi \sample \operatorname{S}_n$.
    	  \State Apply $\pi$ to permute the elements of the tuple $(\mu_1, \dots, \mu_N)$ to obtain $(\mu_1', \dots, \mu_N')$.
   	  \For{$i \in [1,N]$}
    	  \State Sample $\cptp U_i' \sample \mathcal{U}$ from the Haar measure on $\mathcal{U}$.
      \State Sample $n_i \sample \Pois(\mu_i')$.
      \If{$n_i = 0$}
      \State Send vacuum to the Receiver.
      \ElsIf{$n_i = 1$}
      \State Prepare the quantum state $\cptp U_i' (\rho_0)$ and send it to the Receiver.
      \Else %
      \State Send $(n_i, U_i')$ to the Receiver.
      \EndIf
      \EndFor
      \EndProcedure

      \Procedure{\textbf{Emulation of the Decoy Estimation}}{}
      \State Receive $I'$ from the Receiver of size $K$.
      \State Undo the permutation, i.e. let $I \gets \{ \pi^{-1}(i) | i \in I' \}$.
      \State Perform the decoy estimation routine by running $\mathcal{B}$ on input $I$.
      \If{$\mathcal{B}$ does not returns $\Acc$}
      \State Send $\Abort$ to all interfaces.
      \EndIf
      \State Relabel the kept state from $1$ to $K$. Sample a random permutation $\sigma \sample \operatorname{S}_K$.
      \For{$j \in [1,K]$}
      \State Let $\tilde{\cptp U}_j = \cptp U_{\sigma(j)} \cptp U_j'^\dagger$.
      \EndFor
      \State Send the corrections $\sigma$, $(\tilde{\cptp U}_j)_{j\in [1, K]}$ to the Receiver.
      \EndProcedure
    \end{algorithmic}
  \end{game}
  \end{figure}
	
	From the outside, Game~\ref{game:real-ideal-hybrid2} is perfectly indistinguishable from the composition of the Sender's instructions in the protocol in the real world with the $N$ instances of the $\mathsf{WCPGenerator}$ resource, which concludes the proof of indistinguishability of the ideal and the real world.
\end{proof}
 
\section{Proof of Theorem~\ref{thm:instantiation}}\label{proof:ac-security}

We prove here the results from Section~\ref{sec:example}, restated below.

\setcounter{temp-count}{\value{theorem}}
\setcounter{theorem}{\value{count:instantiation}-1}
\begin{theorem}[Correctness and Security Errors]
    Given $\Delta_0 > 0$ in Algorithm~$\mathcal B$, $\delta > 0$ such that $K = ((2-e^{-\eta\nu}-e^{-\eta\nu'})/2 - \delta)N$ and $C = \max(c, c')$, the correctness error $\varepsilon_{\text{corr}}$ satisfies
  \begin{align}
    \varepsilon_{\text{corr}} \leq \exp(-\delta^2 N) + \exp(-\frac{\Delta_0^2 (bc' - b'c)^2}{4 C^2}N)
  \end{align}
while the security error is negligible in $N$ whenever there are additional constants $\Delta_0', \Delta_0'' > 0$ such that
\begin{align}
\Delta_0 + \Delta_0' + \frac{c'}{bc'-b'c}\Delta_0''  = \frac{c'(1-e^{-\eta\nu}) - c(1-e^{-\eta\nu'}) - c'(1-a-b-c)}{bc'-b'c}.
\end{align}
\end{theorem}

\setcounter{theorem}{\value{temp-count}}

\begin{proof}[Proof of Correctness]
  In this proof, we consider that both players follow their protocol instructions honestly. Following Theorem~\ref{thm:decoy-realization}, the correctness error is upper bounded by the probability of Game~\ref{game:correctness} to return $\Abort$. This can happen if $|I|<K$ or if Algorithm~$\mathcal B$ returns $\Abort$. The probability of each event is upper bounded separately.
  
  We start by defining the random variable $P = D_0 + D_1 + D_2 + D_{\geq 3}$ where $D_i$ (resp. $D_{\geq i}$) is the number of acknowledged pulses with intensity $\nu$ with $i$ photons (resp. $\geq i$ photons). Likewise, we define $P' = D_0' + D_1' + D_2' + D_{\geq 3}'$ for pulses with intensity $\nu'$. 
\paragraph{Evaluating the probability that $\mathbf{|I|<K}$.}
By construction, with $I$ defined in Game~\ref{game:correctness}, we have $|I| = P + P'$. This is because $D_0 = 0$ and $D_0' = 0$ in the honest case. Now, we note that $P$ is a binomial distribution with parameters $N/2$ and $1-e^{-\eta\nu}$ and likewise for $P'$  with $\nu$ replaced by $\nu'$. Therefore, for $\delta = \frac{2 - e^{-\eta\nu} - e^{-\eta\nu'}}{2} - \frac{K}{N}$, using Hoeffding's bound and the union bound we obtain
\begin{equation}
  \Pr(|I| < K ) \leq 2\exp(-\delta^2 N).
\end{equation}

\paragraph{Evaluating the probability that Algorithm $\mathcal B$ returns $\Abort$.}
Algorithm $\mathcal B$ returns abort whenever $T < t - \Delta_0 \frac{N}{2}$ where $T$ is given by Eq.~\ref{eq:T} and $t$ by Eq.~\ref{eq:t}. This happens with the following probability:
\begin{align}
  \Pr(T< t - \Delta_0 \frac{N}{2})
  & = \Pr(c'P - cP' < \frac{N}{2} \left(c'(1-e^{-\eta\nu}) - c(1-e^{-\eta\nu' }) - \Delta_0 (bc' - b'c)\right)) \\
  & \leq 2 \max \left\{\Pr((1-e^{-\eta\nu}) - \frac{2P}{N} > \frac{\Delta_0(bc' - b'c)}{2c'}), \right.\\
  &\qquad\qquad \left. \Pr(\frac{2P'}{N} -(1-e^{-\eta\nu}) > \frac{\Delta_0(bc' - b'c)}{2c})\right\} \\
  & \leq 2 \max \left\{\exp(-\frac{\Delta_0^2(bc' - b'c)^2}{4{c'}^2}N),  \exp(-\frac{\Delta_0^2(bc' - b'c)^2}{4{c}^2}N)\right\},
\end{align}
where again we applied Hoeffding's bound and the union bound, taking into account that $\nu < \nu'$ implies $bc' - b'c > 0$.

\paragraph{Combining both cases.}
Using one last time the union bound to combine the two cases above, we obtain that the correctness error $\varepsilon_{\text{corr}}$ satisfies:
  \begin{align}
    \varepsilon_{\text{corr}} \leq 2 \exp(-\delta^2 N) + 2 \exp(-\frac{\Delta_0^2 (bc' - b'c)^2}{4 C^2}N),
  \end{align}
  for $C = \max(c, c')$. 
  This implies that for any choice of constant $\delta>0$ such that $K = ((2-e^{-\eta\nu}-e^{-\eta\nu'})/2 - \delta)N$ and for any choice of constant $\Delta_0 > 0$, the probability of $\Abort$ of Game~\ref{game:correctness} is negligible in $N$.

\end{proof}

\begin{proof}[Proof of Soundness]
For security, we need to evaluate the probability that Game~\ref{game:security} returns $\Fail$ against a cheating Receiver. Compared to an honest participant, the adversary can choose to remove all losses in the fibre, report some pulses as empty while they are in fact non-empty. Its goal is to only report pulses with more than two photons to the Sender so that it may extract information about the secret parameters.

\paragraph{Proof strategy.}
Game~\ref{game:security} returns $\Fail$ if the set $I$ contains more than $K$ elements, $D_0 + D_1 + D_0' + D_1' = 0$ and Algorithm~$\mathcal B$ returns $\Acc$, i.e.~$T \geq t - \Delta_0 \frac{N}{2}$. We therefore need to bound the following probability. If we denote $\mathsf{Cheat}$ the event $D_0 + D_1 + D_0' + D_1' = 0$, we can write the failure probability as:
\begin{align}
  \varepsilon_{\text{sec}} &= \Pr(|I| \geq K \wedge \Acc \wedge \mathsf{Cheat}) \\
  &\leq \Pr(|I| \geq K \wedge \Acc \mid \mathsf{Cheat})\\
  &\leq \Pr(\Acc \mid |I| \geq K \wedge \mathsf{Cheat})\Pr(|I| \geq K \mid \mathsf{Cheat}),
\end{align}
since the decision of whether to cheat rests with the adversary.

We will start by bounding the probability $\Pr(|I| \geq K \mid \mathsf{Cheat})$ that enough pulses are acknowledged by the adversary even after rejecting all the pulses with $0$ or $1$ photons. In particular, we will exhibit a condition on $\eta, \nu, \nu'$ under which there is an overwhelming probability that there are not enough pulses with only two photons or more to pass the requirement $|I| \geq K$. If this condition is not satisfied, we upper-bound that probability by $1$.

In the second part of the proof, we bound the probability that the check from algorithm $\mathcal{B}$ succeeds. The condition $T \geq t - \Delta_0\frac{N}{2}$ required by Algorithm $\mathcal{B}$ to return $\Acc$ can be expanded as:
\begin{align}
T \geq t - \Delta_0\frac{N}{2} \Leftrightarrow c'P - cP' \geq (c'(1-e^{-\eta\nu}) - c(1-e^{-\eta\nu'}) - \Delta_0)(bc'-b'c)\frac{N}{2}.
\end{align}
We simplify this by carefully restricting the domain of $P$ and $P'$ and proving that the probability that they fall outside of this restricted domain is negligible. This will yield an expression of these random variables that is much easier to bound by a negligible quantity. Outside of the domain, the probability that the adversary succeeds in cheating is taken to be $1$.

We again define the random variables $P = D_0 + D_1 + D_2 + D_{\geq 3}$ and $P'= D_0' + D_1' + D_2' + D_{\geq 3}'$ as in the previous proof. We again have $|I| = P + P'$. The condition $D_0 + D_1 + D_0' + D_1' = 0$ implies the following simplification
\begin{align}
  P & = D_2 + D_{\geq 3} \\
  P' & = D_2' + D_{\geq 3}'.
\end{align}

\paragraph{When does the adversary have enough pulses to cheat?}
Our aim in this part of the proof is to bound the probability $\Pr(|I| \geq K \mid \mathsf{Cheat})$. We have $K = ((2-e^{-\eta\nu}-e^{-\eta\nu'})/2 - \delta)N$ from the correctness proof above, and $|I| = P + P' = D_2 + D_2' + D_{\geq 3} + D_{\geq 3}' = D_{\geq 2} + D_{\geq 2}'$. We define $C_{\geq i}$ (respectively $C_{\geq i}'$) to be the random variable counting the number of pulses of intensity $\nu$ (respectively $\nu'$) with $i$ or more photons. Since $D_{\geq 2} \leq C_{\geq 2}$ and $D_{\geq 2}' \leq C_{\geq 2}'$ we have that:
\begin{align}
\Pr(|I| \geq K \mid \mathsf{Cheat}) \leq \Pr(C_{\geq 2} + C_{\geq 2}' \geq K).
\end{align}

$C_{\geq 2}$ follows a binomial distribution with parameters $N/2$ and $(1-a-b)$, and similarly for $C_{\geq 2}'$ with parameters $N/2$ and $(1-a'-b')$.
Let $\mathbf{C}_{\geq 2} = C_{\geq 2} + C_{\geq 2}'$. For any constant $\Gamma > 0$, using Hoeffding's bound and the union bound we obtain
\begin{equation}
  \Pr(\mathbf{C}_{\geq 2} > \left(\frac{2 - a - b - a' - b'}{2} + \Gamma\right)N) \leq 2\exp(-\Gamma^2 N).
\end{equation}

Therefore, if we have
\begin{align}
K \geq \left(\frac{2 - a - b - a' - b'}{2} + \Gamma\right)N &\Leftrightarrow \frac{2-e^{-\eta\nu}-e^{-\eta\nu'}}{2} - \delta \geq \frac{2 - a - b - a' - b'}{2} + \Gamma\\
&\Leftrightarrow 2(\Gamma + \delta) \leq a + b + a' + b' - (a^\eta + {a'}^\eta)
\end{align}
then, the adversary can cheat with at most negligible probability $2\exp(-\Gamma^2 N)$. 
This is minimised by taking $\Gamma = \frac{a + b + a' + b' - (a^\eta + {a'}^\eta)}{2} - \delta$.
Therefore the probability $\Pr(|I| \geq K \mid \mathsf{Cheat})$ can be upper-bounded by:
\begin{align}
\Pr(|I| \geq K \mid \mathsf{Cheat}) \leq 
\begin{dcases*}
2\exp(-\Gamma^2 N) 
   & if  $0 < \delta < \frac{a + b + a' + b' - (a^\eta + {a'}^\eta)}{2}$\,, \\[1ex]
1
   & otherwise\,.
\end{dcases*}
\end{align}
for $\Gamma = \frac{a + b + a' + b' - (a^\eta + {a'}^\eta)}{2} - \delta$.

\paragraph{Bounding the success probability of $\mathcal{B}$.}
Our goal here is to bound the probability $\Pr(\Acc \mid |I| \geq K \wedge \mathsf{Cheat})$. Recall that in this case $\Acc$ is equivalent to satisfying $c'P - cP' \geq (c'(1-e^{-\eta\nu}) - c(1-e^{-\eta\nu'}) - \Delta_0)(bc'-b'c)\frac{N}{2}$, for $P = D_2 + D_{\geq 3}$ and $P' = D_2' + D_{\geq 3}'$.

\textit{Simplifying $\mathbf{P}$ and $\mathbf{P'}$.}
We focus on $D_2$ and $D_2'$. Let $\mathbf{D}_2 = D_2 + D_2'$. Let $C_2$ and $C_2'$ be the number of pulses of intensity respectively $\nu$ and $\nu'$ which contain two photons, and $\mathbf{C}_2 = C_2 + C_2'$. Because $D_2 \sim H(\mathbf{C}_2, \mathbf{D}_2, C_2)$ and $D_2' \sim H(\mathbf{C}_2,\mathbf{D}_2,C_2')$ follow hypergeometric distributions, we have that $D_2$ and $D_2'$ converge toward $\frac{\mathbf{D}_2}{\mathbf{C}_2}  C_2$ and $\frac{\mathbf{D}_2}{\mathbf{C}_2}  C_2'$ respectively. That is, for any constant $\delta_0, \delta'_0 > 0$,
\begin{align}
\label{eq:d-2-c-2}
  \Pr(\abs{\frac{D_2}{C_2} - \frac{\mathbf{D}_2}{\mathbf{C}_2}} \geq \delta_0) & \leq 2 \exp(-2 \delta_0^2 C_2) \\ 
  \Pr(\abs{\frac{D_2'}{C_2'} - \frac{\mathbf{D}_2}{\mathbf{C}_2} } \geq \delta'_0) & \leq 2 \exp(-2 {\delta'}^2_0 C_2')
\end{align}
We now have to analyse how the rhs of these inequalities behave.

Then, for constants $\gamma_0, \gamma_0' > 0$, we have the following tail bounds:
\begin{align}
\label{eq:c-2-n}
  \Pr(\abs{C_2\frac{2}{N} - c}  > \gamma_0) & \leq 2\exp(-\gamma_0^2 N) \\
  \Pr(\abs{C_2'\frac{2}{N} - c'} > \gamma_0') & \leq 2\exp(-{\gamma_0'}^2 N).
\end{align}
Combining the bounds from Equations~\eqref{eq:d-2-c-2} and \eqref{eq:c-2-n} (and the corresponding ones for $D_2'$), we get that we can express $D_2$ and $D_2'$ as
\begin{align}
  D_2 & = \left(\frac{\mathbf{D}_2}{\mathbf{C}_2} + \delta\right)\left(c + \gamma\right) \frac{N}{2}\\
  D_2' & = \left(\frac{\mathbf{D}_2}{\mathbf{C}_2} + \delta' \right)\left(c' + \gamma'\right) \frac{N}{2}, 
\end{align}
with
\begin{align}
   \Pr(|\gamma| > \gamma_0)   & \leq 2 \exp(-\gamma_0^2 N), \\
   \Pr(|\delta| > \delta_0)   & \leq  2 \times 2 \max\left\{ \exp(-\gamma_0^2 N), \exp(-\delta_0^2 (c-\gamma_0) N) \right\}, \\
   \Pr(|\gamma'| > \gamma'_0) & \leq 2 \exp(-{\gamma'_0}^2 N), \\
   \Pr(|\delta'| > \delta_0') & \leq 2 \times 2 \max\left\{ \exp(-{\gamma'_0}^2 N), \exp(-{\delta'_0}^2 (c-\gamma_0') N) \right\}.
\end{align}
Indeed, 
\begin{align}
\Pr(\abs{\delta} > \delta_0)  & = \overbrace{\Pr(|\delta| > \delta_0 \vert \abs{\gamma} > \gamma_0)}^{\strut \leq 1} \overbrace{\Pr(\abs{\gamma} > \gamma_0)}^{\strut \leq 2 \exp(-\gamma_0^2 N)} \\ &+ \underbrace{\Pr(|\delta| > \delta_0 \vert \abs{\gamma} \leq \gamma_0)}_{\leq 2 \underset{C_2\,\rm{s.t.}\,\abs{\gamma} \leq \gamma_0 }{\max} \exp(-2 \delta_0^2 C_2)} \underbrace{\Pr(\abs{\gamma} \leq \gamma_0)}_{\strut \leq 1} \nonumber \\
\Pr(\abs{\delta} > \delta_0)  & \leq 2 \exp(-\gamma_0^2 N) + 2 \exp(-\delta_0^2 (c-\gamma_0) N) \\
\Pr(\abs{\delta} > \delta_0)  & \leq 2 \times 2 \max \left\{\exp(-\gamma_0^2 N), \exp(-\delta_0^2 (c-\gamma_0) N) \right\}
\end{align}

\textit{Domain restriction.}
Hence, if we denote $\mathsf{Dom}_0$ the event $|\gamma| \leq \gamma_0 \wedge |\delta| \leq \delta_0 \wedge |\gamma'| \leq \gamma_0' \wedge |\delta'| \leq \delta'_0$, the four conditions above imply
\begin{align}
\label{eq:bound-domain}
  & \Pr(\neg \mathsf{Dom}_0) \\
  & \quad \leq 32 \max\left\{ \exp(-\gamma_0^2 N), \exp(-\delta_0^2 (c-\gamma_0) N), \exp(-{\gamma'_0}^2 N), \exp(-{\delta'_0}^2 (c'-\gamma_0') N) \right\}.\nonumber 
\end{align}
Since
\begin{align*}
\label{eq:comb}
\Pr(\Acc \mid |I| \geq K \wedge \mathsf{Cheat}) &= \Pr(\Acc \mid |I| \geq K \wedge \mathsf{Cheat} \wedge \neg \mathsf{Dom}_0)\Pr(\neg \mathsf{Dom}_0) \\
&\quad + \Pr(\Acc \mid |I| \geq K \wedge \mathsf{Cheat} \wedge \mathsf{Dom}_0)\Pr(\mathsf{Dom}_0)\\
&\leq \Pr(\neg \mathsf{Dom}_0) + \Pr(\Acc \mid |I| \geq K \wedge \mathsf{Cheat} \wedge \mathsf{Dom}_0)
\end{align*}
with $\Pr(\neg \mathsf{Dom}_0)$ negligible in $N$. We now only need to analyse the probability of $\Acc$ in the $\mathsf{Dom}_0$ case with a cheating Receiver and $|I| \geq K$.

\textit{Bounding the restricted probability.}
We rewrite the $\Acc$ condition for Algorithm~$\mathcal B$ as follows using our simplifications for $P$ and $P'$
\begin{align}
  & T \geq t - \Delta_0\frac{N}{2} \\
  \Leftrightarrow
  & \left( c'(D_2 + D_{\geq 3}) - c (D_2' + D_{\geq 3}')\right)\frac{2}{N}\nonumber \\
  & \qquad \geq \left(c'(1-e^{-\eta\nu}) - c(1-e^{-\eta\nu'})\right) - \Delta_0(bc' - b'c) \\
  \Leftrightarrow
  & \left(cc'(\delta - \delta') + \frac{\mathbf{D}_2}{\mathbf{C}_2}(c'\gamma - c\gamma') + (c'\gamma\delta - c\gamma'\delta')\right) + (c'D_{\geq 3} - cD_{\geq 3}') \frac{2}{N} \nonumber \\
  & \qquad \geq \left(c'(1-e^{-\eta\nu}) - c(1-e^{-\eta\nu'})\right) - \Delta_0(bc' - b'c)
\end{align}
We can now simplify the lhs even more by upper-bounding it as follows
\begin{align}
  & \left(cc'(\delta - \delta') + c'\gamma \left(\frac{\mathbf{D}_2}{\mathbf{C}_2}+ \delta\right) - c\gamma'\left(\frac{\mathbf{D}_2}{\mathbf{C}_2}+\delta'\right)\right) + c'C_{\geq 3}\frac{2}{N} \nonumber \\
  & \qquad \geq \left(cc'(\delta - \delta') + \frac{\mathbf{D}_2}{\mathbf{C}_2}(c'\gamma - c\gamma') + (c'\gamma\delta - c\gamma'\delta')\right) + (c'D_{\geq 3} - cD_{\geq 3}') \frac{2}{N}.
\end{align}
Because $|\gamma| \leq \gamma_0$, $|\gamma'| \leq \gamma_0'$, $|\delta| \leq \delta_0$, $|\delta'| \leq\delta_0'$ and $0\leq \frac{\mathbf{D}_2}{\mathbf{C}_2}\leq 1$ we have
\begin{equation}
  cc'(\delta_0 + \delta_0') + c'\gamma_0 (1+ \delta_0) + c\gamma_0'(1+\delta_0') \geq cc'(\delta - \delta') + c'\gamma \left(\frac{\mathbf{D}_2}{\mathbf{C}_2}+ \delta\right) - c\gamma'\left(\frac{\mathbf{D}_2}{\mathbf{C}_2}+\delta'\right). 
\end{equation}
Defining the new constant $\Delta_0'$ via $(bc'-b'c)\Delta_0' = cc'(\delta_0 + \delta_0') + c'\gamma_0 (1+ \delta_0) + c\gamma_0'(1+\delta_0')$ gives
\begin{equation}
(bc'-b'c)\Delta_0' \geq   cc'(\delta - \delta') + c'\gamma \left(\frac{\mathbf{D}_2}{\mathbf{C}_2}+ \delta\right) - c\gamma'\left(\frac{\mathbf{D}_2}{\mathbf{C}_2}+\delta'\right)
\end{equation}
so that,
\begin{align}
  & \Pr(\Acc \mid |I| \geq K \wedge  \mathsf{Cheat} \wedge \mathsf{Dom}_0) \nonumber \\
  & \quad \leq \Pr(c'C_{\geq 3}\frac{2}{N} + (\Delta_0 + \Delta_0')(bc' - b'c)  \geq c'(1-e^{-\eta\nu}) - c(1-e^{-\eta\nu'})).
\end{align}
Using the fact that $C_{\geq 3}$ follows a binomial distribution with parameters $N/2$ and $(1-a-b-c)$, we have that for a constant $\Delta_0''>0$,
\begin{equation}
  \Pr(C_{\geq 3}\frac{2}{N} - (1-a-b-c) > \Delta_0'' ) \leq \exp(-{\Delta_0''}^2 N).\label{eq:scaling}
\end{equation}

\paragraph{Conclusion.}
The result of the bounds derived above is that we obtain a negligible bound on the cheating probability with parameters $\nu, \nu', \Delta_0, \Delta_0', \Delta_0'' > 0$ if the following constraint is satisfied
\begin{align}\label{eq:error}
\Delta_0'' = \frac{1}{c'}\left(c'(1-e^{-\eta\nu}) - c(1-e^{-\eta\nu'}) - (\Delta_0 + \Delta_0')(bc'-b'c)\right) - (1-a-b-c)
\end{align}

Then we conclude that the probability that $|I| > K$ and that Algorithm~$\mathcal B$ accepts conditioned on $D_0 + D_1 + D_0' + D_1' = 0$ and $\mathsf{Dom}_0$ is negligible. Furthermore we proved earlier that the probability of being outside the domain is negligible. Therefore the probability that $|I| > K$ and Algorithm~$\mathcal B$ returns $\Acc$ and $D_0 + D_1 + D_0' + D_1' = 0$ is negligible in $N$.

\end{proof}

\begin{proof}[Proof of AC Security]

The complete AC security error of the protocol corresponds to the sum of the correctness and security errors. Therefore we summarise and combine the conditions obtained from both of the proofs above. This can be expressed as minimising the following function
\begin{align}
\varepsilon_{{\rm AC}} &= \varepsilon_{\rm corr} + \varepsilon_{\rm sec}\\
&= 2\exp(-\delta^2 N) + 2\exp(-\frac{\Delta_0^2 (bc' - b'c)^2}{4 C^2}N) + P_{|I|, K}\left(M_{\delta_0, \delta_0', \gamma_0, \gamma_0'} + \exp(-{\Delta_0''}^2 N)
\right),
\end{align}
with 
\begin{align}
M_{\delta_0, \delta_0', \gamma_0, \gamma_0'} &= \max\left\{ \exp(-\gamma_0^2 N), \exp(-\delta_0^2 (c-\gamma_0) N), \exp(-{\gamma'_0}^2 N), \exp(-{\delta'_0}^2 (c'-\gamma_0') N) \right\} \label{eq:epsilon_AC}\\
P_{|I|, K} &= 
\begin{dcases*}
2\exp(-\left(\frac{a + b + a' + b' - (a^\eta + {a'}^\eta)}{2} - \delta\right)^2 N) 
   & if  $0 < \delta < \frac{a + b + a' + b' - (a^\eta + {a'}^\eta)}{2}$\,, \\[1ex]
1
   & otherwise\,.
\end{dcases*}
\end{align}

under the following set of constraints
\begin{align}
&0 < \delta < \frac{2 - a^\eta - {a'}^\eta}{2}; 0 < \Delta_0; 0 < \delta_0; 0 < \delta_0'; 0 < \gamma_0 ; 0 < \gamma_0' \label{eq:constr_deltas} \\
&0 < \Delta_0' = \frac{1}{bc' - b'c}(cc'(\delta_0 + \delta_0') + c'\gamma_0 (1 + \delta_0) + c\gamma_0'(1 + \delta_0')) \label{eq:constr_delta_p}\\
&0 < \Delta_0'' = \frac{1}{c'}\left(c'(1-a^\eta) - c(1-{a'}^\eta) - (\Delta_0 + \Delta_0')(bc'-b'c)\right) - (1-a-b-c) \label{eq:constr_delta_pp} \\
&C = \max\left\{c, c'\right\}. 
\end{align}

\end{proof}

\section{Proof of Theorem~\ref{thm:scaling}}
\label{sec:app-scaling}
\setcounter{theorem}{\value{count:scaling}-1}
\begin{theorem}[Scaling for $\eta \rightarrow 0$]
For two pulses of intensities $\nu, \nu'$  with $\nu  =  \alpha\nu'$ and $0<\alpha<1$, the number of pulses needed to obtain a given correctness and security error scales as $1/\eta^{2}$ for $\eta \rightarrow 0$ .
\end{theorem}

\setcounter{theorem}{\value{temp-count}}

\begin{proof}[Scaling]
  The strategy for proving the scaling is first to ensure that the constraints for the validity of the bound on $\varepsilon_{\mathrm{AC}}$ can be met. Then, it needs to be checked that $\varepsilon_{\mathrm{AC}}$ can be made arbitrarily small as $\eta \rightarrow 0$.

  We start with Equation~\eqref{eq:constr_delta_pp}, and expand it to the lowest order in $\eta$ and setting $\nu = \alpha \nu'$ for $0<\alpha<1$. This gives:
  \begin{equation}
    \Delta_0'' \simeq \left(\eta - (\Delta_0 + \Delta_0')\right)\nu'\alpha(1-\alpha).
  \end{equation}  
  For small $\varepsilon_{\mathrm{AC}}$, we can thus keep $\Delta_0'' > 0$ by ensuring that $\Delta_0, \Delta_0' \lesssim \frac{\eta}{3}$ which would them imply $\Delta_0'' \lesssim \eta\nu'\alpha(1-\alpha)/3$.  In turn to ensure that $\Delta_0'\lesssim \frac{\eta}{3}$, we realize that $\Delta_0'$ is independent of $\eta$. This implies that by choosing $\delta_0, \delta_0', \gamma_0, \gamma_0' \propto \eta$, this can be ensured.  For $\Delta_0$, we observe that it can be freely set to $\eta/3$ as the only constraint (Equation~\eqref{eq:constr_deltas}) is $0 < \Delta_0$.  Finally, by setting $\delta \propto \eta$, we can achieve $\delta < (2-a^{\eta}-a'^\eta)/2$.

  Crucially, the above constraints allow to keep $\nu'$ constant so that each term in Equation~\eqref{eq:epsilon_AC} can be kept below a target value for $\eta \rightarrow 0$ as long as $N \propto \frac{1}{\eta^{2}}$, which concludes the proof.
\end{proof}

\end{document}